\documentclass[11pt,a4paper]{article}

\usepackage[utf8x]{inputenc}
\usepackage{amsmath,ifthen,color,eurosym}

\usepackage[]{datetime}
\usepackage{wrapfig}
\usepackage{hyperref}
\usepackage{latexsym,amsthm,amsmath,amssymb,url}
\usepackage{fancyhdr}
\usepackage{graphicx}
\usepackage{comment}
\usepackage{stmaryrd}
\usepackage{enumerate}
\usepackage{cite}
\usepackage{tikz}

\usepackage{todonotes}
\newtheorem{theorem}{Theorem}
\newtheorem{proposition}{Proposition}
\newtheorem{lemma}{Lemma}

\newtheorem{observation}{Observation}
\newtheorem{claim}{Claim}

\renewcommand{\int}{{\bf int}}

\newcommand{\yes}{\mbox{\sc yes}}
\newcommand{\no}{\mbox{\sc no}}

\newcommand{\allocation}{{\sc List  Allocation}}

\newcommand{\la}{{\sc  LA}}
\newcommand{\cla}{{\sc  CLA}}
\newcommand{\hcla}{{\sc  HCLA}}
\newcommand{\shcla}{{\sc  S-HCLA}}

 \usepackage{todonotes}
\newcommand{\remove}[1]{}

\usepackage[lined,boxed,commentsnumbered]{algorithm2e}

\usepackage{graphicx}

\newenvironment{proofofclaim}{\noindent \textsc{Proof of the Claim:}}{\hfill$\Diamond$\medskip}

\usepackage{cite}

\begin{document}

\date{\empty}

\title{Parameterized Algorithms for Min-Max Multiway Cut and List Digraph Homomorphism\footnote{
The third author was co-financed by the European Union (European Social Fund ESF) and Greek national funds through the Operational Program ``Education and Lifelong Learning'' of the National Strategic Reference Framework (NSRF), Research Funding Program: ARISTEIA~II.}~\footnote{Emails:
\texttt{eunjungkim78@gmail.com}, \texttt{\{Christophe.Paul,Ignasi.Sau\}@lirmm.fr}, and \texttt{sedthilk@thilikos.info}
}~\footnote{An extended abstract of this work will appear in the \emph{Proceedings of the 10th International Symposium on Parameterized and Exact Computation (IPEC), Patras, Greece, September 2015}.}}

\author{Eunjung Kim\thanks{CNRS, LAMSADE, Paris, France.} \and Christophe Paul\thanks{AlGCo project-team, CNRS, LIRMM, Montpellier, France} \and Ignasi Sau$^{\P}$ \and\\ Dimitrios M. Thilikos$^{\P}$\thanks{Department of Mathematics, University of Athens, Athens, Greece}~\thanks{Computer Technology Institute \& Press  ``Diophantus'', Patras, Greece}}

\maketitle

\vspace{0mm}
\begin{abstract}\medskip\medskip

\noindent In this paper we design {\sf FPT}-algorithms for two parameterized problems.
The first is   \textsc{List Digraph Homomorphism}:
given two digraphs $G$ and $H$ and a list of allowed vertices of $H$ for
every vertex of $G$, the question is
whether there exists a homomorphism from $G$ to $H$
respecting the list constraints.
The second problem is a variant of \textsc{Multiway Cut}, namely \textsc{Min-Max Multiway Cut}:
given a graph $G$, a non-negative integer $\ell$,
 and a set $T$ of $r$ terminals, the question is
whether we can partition the vertices of $G$ into $r$
parts such that (a) each part contains one terminal and
(b) there are at most $\ell$ edges with only one endpoint in this part.
We parameterize \textsc{List Digraph Homomorphism}
by the number $w$ of edges of $G$ that are mapped to non-loop edges of $H$
and we give a time
$2^{O(\ell\cdot\log h+\ell^2\cdot \log \ell)}\cdot n^{4}\cdot \log n$ algorithm,
where $h$ is the order of the host graph $H$.
We also prove that
\textsc{Min-Max Multiway Cut} can be solved in time $2^{O((\ell r)^2\log \ell r)}\cdot n^{4}\cdot \log n$.
Our approach introduces a general problem, called  {\sc List Allocation}, whose expressive power
permits the design of parameterized  reductions of both aforementioned problems to it.
Then our results are based on an {\sf FPT}-algorithm for the  {\sc List Allocation} problem
that is designed using a  suitable adaptation of the {\em randomized contractions} technique
(introduced by [Chitnis, Cygan, Hajiaghayi, Pilipczuk, and Pilipczuk, FOCS 2012]).
\end{abstract}


\section{Introduction}
\label{sec:intro}

The {\sc Multiway Cut} problem asks, given a graph $G$, a set of $r$ terminals $T$, and a non-negative integer $\ell$,
whether it is possible to partition $V(G)$ into $r$ parts such that each part contains exactly one of the terminals of $T$ and
there are at most $\ell$ edges between different parts (i.e., at most $\ell$ crossing edges).
In the special case where $|T|=2$, this gives the celebrated {\sc Minimum Cut}
problem, which is polynomially solvable~\cite{StWa97}. In general, when there is no restriction on the number
of terminals, the  {\sc Multiway Cut} problem  is {\sc NP}-complete~\cite{DahlhausJPSY94thec}
and a lot of research has been devoted
to the study of this problem and its generalizations, including several classic results
on its polynomial  approximability~\cite{GoemansW95impr,ChekuriGN06thes,EvenNSS98appro,GargVY04mult,KargerKSTY99roun,Ravi02appr}.

More recently, special attention to the  {\sc Multiway Cut} problem was given from the parameterized complexity
point of view. The existence of an {\sf FPT}-algorithm for
{\sc Multiway Cut} (when parameterized by $\ell$), i.e., an $f(\ell)\cdot n^{O(1)}$-step algorithm,
had been a long-standing open problem. This question was answered positively
by Marx in~\cite{Marx06para} with the use of  the {\em important separators technique} which was
also used for the design of {\sf FPT}-algorithms for several other problems
such as  {\sc Directed Multiway Cut} ~\cite{ChitnisHM12fixe},
{\sc Vertex Multicut}, and {\sc Edge Multicut}~\cite{MarxR14fixe}.
%,ChenLL09anim}
This  technique has been  extended to the powerful framework of {\sl randomized
contractions technique}, introduced in~\cite{ChitnisCHPP12desi}. This   made it
possible to design  {\sf FPT}-algorithms for
several other problems such as {\sc  Unique Label Cover}, {\sc Steiner Cut},
{\sc Edge/Vertex Multiway Cut-Uncut}.
We stress that this technique is quite versatile.

In this paper we use
it in order
to design {\sf FPT}-algorithms for parameterizations of two problems that do not seem to be directly related to each other:
the   {\sc Min-Max-Multiway Cut} problem~\cite{SvitkinaT04minm} and the {\sc List Digraph Homomorphism} problem.
%There are  \textsc{Min-Max Multiway Cut} and \textsc{Arc-specified List Digraph Homomorphism}.
%The first is a known variant of  \textsc{Multiway Cut} parameterized by $w$ and $r$ and the second
%is a specific parameterization of the  {\sc List Digraph Homomorphism} problem.

\vspace{0mm}\subsection{\textbf{{Min-Max-Multiway Cut}}} In the  {\sc Multiway Cut} problem the parameter $\ell$ bounds
the total number of crossing edges (i.e., edges with endpoints in different parts).
Svitkina and Tardos~\cite{SvitkinaT04minm} considered a ``min-max'' variant of this problem, namely the {\sc Min-Max-Multiway Cut},
where $\ell$ bounds  the maximum number of outgoing edges of the parts\footnote{Notice that under this viewpoint
{\sc Multiway Cut} can be seen as {\sc Min-Sum-Multiway Cut}.}.
In~\cite{SvitkinaT04minm}, it was proved that  {\sc Min-Max-Multiway Cut} is {\sf NP}-complete even when the number of terminals is $r=4$.
As a consequence of the results in~\cite{SvitkinaT04minm} and~\cite{Racke08opti},
{\sc Min-Max-Multiway Cut} admits
 an $O(\log^2 n)$-approximation algorithm. This was improved
 recently in~\cite{BansalFKMNNS11minm} to a $O((\log n\cdot \log r)^{1/2})$-approximation algorithm.

To our knowledge, nothing is known about the parameterized complexity of this problem.
We prove the following.

 \begin{theorem}
 \label{mmmc}
There exists an algorithm that solves the {\sc Min-Max-Multiway Cut}
problem in  $2^{O((r\ell)^2\log r\ell)}\cdot n^{4}\cdot \log n$ steps,
 i.e., {\sc Min-Max-Multiway Cut} belongs to {\sf FPT}  when parameterized by both $r$ and $\ell$.
\end{theorem}
(Throughout the paper, we use $n=|V(G)|$ when we refer to the
number of vertices of the graph $G$ in the instance of the considered problem.)

%
% The proof is almost the same
% {\sf T-FPT}-reduction to the {\sc List Allocation} problem as the one we described before for {\sc Multiway Cut}. The only  difference is that we now define ${\cal A}$
% so to contain every $\alpha$ such that $$\forall i\in\{1,\ldots,w\},\  \sum_{j\in\{1,\ldots,w\}\setminus\{i\}}\alpha(i,j)\leq w.$$
 %As a result of this, {\sc Min-Max-Multiway Cut} can be solved in
% In particular we p
%
%. The details are presented in Subsection~\ref{io4nfhsod8f}.

\vspace{0mm}\subsection{\textbf{{List Digraph Homomorphism}}}
Given two directed graphs $G$ and $H$, an {\em $H$-homomorphism} of $G$
is a mapping $\chi: V(G)\rightarrow V(H)$ such that if $(x,y)$ is an arc of $G$, then
$(\chi(x),\chi(y))$ is also an arc in $H$.
In the {\sc List Digraph Homomorphism} problem,
we are given two graphs $G$ and $H$ and a list function $\lambda: V(G)\rightarrow 2^{V(H)}$
and we ask whether $G$ has a $H$-homomorphism such that
for every vertex $v$ of $G$, $\chi(v)\in\lambda(v)$.
%
%The {\sc List Digraph Homomorphism} problem,
%we are given two digraphs $G$ and $H$ and a list function $\lambda: V(G)\rightarrow 2^{V(H)}$
%and we ask whether there  is a mapping $\chi: V(G)\rightarrow V(H)$ such
%that if $(x,y)$ is an arc of $G$, then
%$(\chi(x),\chi(y))$ is also an arc in $H$
%and for every vertex $v$ of $G$, $\chi(v)\in\lambda(v)$.
Graph and digraph homomorphisms have been extensively studied
both from the combinatorial and the algorithmic point of view
(see e.g.,~\cite{HellN04grap,BrownMSW08grap,
EgriKLP12thec,FederHH99list,
FederHH03biar}).

Especially for the  {\sc List Digraph Homomorphism} problem,
a dichotomy characterizing  the instantiations of $H$ for which the problem is hard was given in~\cite{HellR11thed} (see also~\cite{EgriHLR14spac}).
Notice that the standard parameterization of {\sc List Digraph Homomorphism}
by the size of the graph $H$ is {\sf para-NP}-complete as it  yields the {\sc 3-Coloring} problem
when $G$ is restricted to be a simple graph and $H=K_{3}$.
A more promising
parameterization of {\sc List Homomorphism}  (for undirected graphs) has been introduced in~\cite{DiazST01colo},
where the parameter is a bound on the number of pre-images of some prescribed set of vertices of $H$ (see also~\cite{DiazST08effi,DiazST04fixe,
MarxOR13find}). Another parameterization, again for the undirected case,
was introduced in~\cite{ChitnisEM13list}, where
the parameter is the number of vertices to be removed from the graph $G$ so that the remaining graph
has a list $H$-homomorphism.

We introduce a new parameterization
of  {\sc List Digraph Homomorphism}  where the parameter
%\ig{we have to unify $w$ and $w$}
is, apart from $h=|V(H)|$, the number of ``crossing edges'', i.e., the edges of $G$ whose endpoints
are mapped to different vertices of $H$.  For this, we enhance the input with an integer $\ell$
and ask for a list digraph homomorphism with at most $\ell$ crossing edges. Clearly, when $\ell=|E(G)|$, this
yields the original problem. We call the new problem  {\sc Bounded  List Digraph Homomorphism} (in short, {\sc BLDH}).
Notice that the fact that {\sc   List Digraph Homomorphism} is {\sf NP}-complete even when $h=3$,
implies that {\sc BLDH} is {\sf para-NP}-complete when parameterized only by $h$. The input of
 {\sc BLDH} is a quadruple $(G,H,\lambda,\ell)$ where $G$ is the guest graph, $H$ is the host graph,   $\lambda: V(G)\rightarrow 2^{V(H)}$ is the list function and $\ell$ is a non-negative integer.
Our next step is to observe  that {\sc BLDH}  is {\sf W}$[1]$-hard, when parameterized only by $\ell$.
To see this consider an input $(G,k)$ of the {\sc Clique} problem
and construct the input $(K,\bar{G},\lambda,\ell)$ where
$K$ is a the complete digraph on $k$ vertices, $\bar{G}$ is the digraph obtained by $G$ by replacing each edge by two opposite direction arcs between  the same endpoints, $\lambda=\{(v,V(G))\mid v\in V(K)\}$, and $\ell=k(k-1)$.
Notice that $(G,k)$ is a {\sc yes}-instance of {\sc Clique} iff $(K,\bar{G},\lambda,\ell)$ is a {\sc yes}-instance
of {\sc BLDH}.

We conclude that when  {\sc BLDH} is parameterized by $\ell$ or $h$ only, then
one may not expect it to be fixed parameter tractable.
This means that the parameterization of  {\sc BLDH} by $h$ and $\ell$ is meaningful
to consider. Our result is the following.
%
%Our next step is to show that {\sc BLDH}, when parameterized only by $h=|V(H)|$ is {\sf para-NP}-complete
%even when $G$ is a series-parallel graph.
%For this, it is enough to prove that {\sc BLDH} is {\sf NP}-hard on series-parallel graphs for every fixed $r=3$.
%For this, we present a simple reduction from the \textsc{Partition} problem. Given an instance $I=\{a_1, \ldots, a_t\}$ of \textsc{Partition}, we build an instance $I'=(G,H,\lambda,\ell)$ of {\sc BLDH} as follows: Let $H=(\{1,2,3\},\{\{1,2\},\{1,3\}\})$
%Let us start to construct a planar graph $G$ from an isolated vertex $r$ where $\lambda(r)=1$.
% For each $i \in [t]$, add to $G$
%a path on $a_i$ new vertices and connect each of these vertices to $r$. Clearly, the resulting graph $G$ is planar.
%For each vertex $v$ different from $r$ we set $\lambda(v)=\{2,3\}$. Finally,  set $\ell=\frac{1}{2}\sum_{i\in[t]}a_{i}$.
%It is clear that
%$I'$  is a  {\sc yes}-instance of \la\ if and only if $I$ is a {\sc yes}-instance of {\sc Partition}.
%Using the standard terminology of parametrized complexity,
%this reduction proves that {\sc LA} is {\sf para-NP}-complete
%when parameterized by $r$.

%
\begin{theorem}
\label{ldh}
There exists an algorithm that solves the  {\sc Bounded List Digraph Homomorphism}
problem in   $2^{O(\ell\cdot\log h+\ell^2\cdot \log \ell)}\cdot n^{4}\cdot \log n$ steps,
i.e.,  {\sc Bounded List Digraph Homomorphism}  belongs to {\sf FPT}
when parameterized by  the number $\ell$ of crossing edges and the number $h$ of vertices of $H$.
\end{theorem}

\vspace{0mm}\subsection{List Allocation}
In order to prove Theorems~\ref{mmmc} and~\ref{ldh}, we prove that both  problems
are Turing {\sf FPT}-reducible%
%\sed{\small Reviewer says:
%In the footnote, this is not how Turing FPT-reductions are formally
%defined. Also the other says that "Give a citation for the standard reduction"}
\footnote{Let ${\bf A}$ and ${\bf B}$ be two parameterized problems.  We say that a parameterized problem ${\bf A}$ is {\em Turing  {\sf FPT}-reducible} to ${\bf B}$ when the existence of an {\sf FPT}-algorithm for ${\bf B}$ implies the existence of an  {\sf FPT}-algorithm for ${\bf A}$. (For brevity, in this paper,  we write  ``{\sf T-FPT}'' instead of  ``Turing {\sf FPT}''.)}
to a single new problem that we call  {\sc List Allocation} (in short, \la).

The {\sc List Allocation} problem is defined as follows:
We are given a graph $G$ and a set of $r$ ``boxes'' indexed  by  numbers from $\{1,\ldots,r\}$.
Each vertex $v$ of $G$ is accompanied with a list $\lambda(v)$ of indices corresponding to the boxes where it is allowed  to be allocated.
Moreover, there is a weight function $\alpha$ assigning to every pair of  different boxes a non-negative integer.
The question is whether there is a way to place  each of the vertices of $G$
into some box of its list such that,  for any two different boxes $i$ and $j,$ the number
of crossing edges between them
is {\sl exactly} $\alpha(i,j)$.

As we easily see in Subsection~\ref{thloca},
%By a straightforward reduction from {\sc Max Cut}, it follows that
 {\sc List Allocation} is {\sf NP}-complete, even when $r=2$.
 Throughout this paper, we parameterize the {\sc List Allocation} problem by the total number $w$ of ``crossing edges''
between different boxes, i.e.,   $w=\sum_{1\leq i<j\leq r} \alpha(i,j)$.

Our main result is that this parameterization of \la\ is in {\sf FPT}.

\begin{theorem}
\label{u7u33ejk}
There exists an algorithm that, given as input  an instance $I=(G,r,\lambda,\alpha)$ of
\allocation, returns an answer to this problem in $2^{O(w^2\cdot \log w)}\cdot n^{4}\cdot \log n$ steps, where $w=\sum_{1\leq i<j\leq r} \alpha(i,j)$.
\end{theorem}
To witness the expressive power of  {\sc List Allocation},
let us first exemplify  why {\sc Multiway Cut}, parameterized by $w$,
is {\sf T-FPT}-reducible to  {\sc List Allocation}. Given an instance  of {\sc Multiway Cut}, we first
discard from its graph  all  the connected components that have at most 1 terminal. Clearly,
this gives an equivalent instance  $(G,T=\{t_{1},\ldots,t_{r}\},w)$ where $r\leq w+1$.

Next, we consider
%the set ${\cal M}$ containing all bijections $\mu: T\rightarrow [l]$ and
  the set ${\cal A}$ containing every
 weight function $\alpha$ such that $\sum_{1\leq i<j\leq r} \alpha(i,j)\leq w.$
 Let also  $\lambda: V(G) \to 2^{[r]}$ be the list function such that if $v=t_{i}\in T$, then $\lambda(v)=\{i\}$, otherwise
 $\lambda(v)=\{1,\ldots,r\}$.
 %For every $\mu\in{\cal M},$ we  set up the list function $\lambda_{\mu}$
  %such that, for every
%vertex $v$,
% $\lambda_{\mu}(v)=\mu(v)$ if $v$ is a terminal, otherwise,  $\lambda_{\mu}(v)=\{1,\ldots,w\}$.
 %
 It is easy to verify
 that $(G,T,w)$ is
a {\sc yes}-instance of {\sc Multiway Cut}  if and only if there exists
%some $\mu\in{\cal M}$ and
some $\alpha\in{\cal A}$ such that
$(G,r,\lambda,\alpha)$ is a {\sc yes}-instance of  {\sc List Allocation}.
This yields the claimed reduction, as $|{\cal A}|$ is clearly bounded by some function of $w$.
This reduction to  the {\sc List Allocation} problem turns out to be quite flexible
and, as we will see in  Subsection~\ref{subsec:minmax2la} (Theorem~\ref{thm:minmax2la}), it can easily be adapted to a
{\sf T-FPT}-reduction of   {\sc Min-Max-Multiway Cut} to {\sc List Allocation}.
The reduction of {\sc Bounded List Digraph Homomorphism}  to {\sc List Allocation}
is more complicated and is described in Subsection~\ref{subsec:hom2la} (Theorem~\ref{thm:hom2la}).
%by suitably
%adapting the definition of $\lambda$ and
%the set ${\cal A}$, we can easily {\sf T-FPT}-reduce more problems
% to  {\sc List Allocation}.
%
Theorem~\ref{u7u33ejk}, together with the aforementioned reductions,
yields Theorems~\ref{mmmc} and~\ref{ldh}.

 \vspace{0mm}\section{Preliminaries and the definition of \allocation}
\label{sec:prelim}

\vspace{0mm}\subsection{\textbf{{Functions and allocations}}}

We use the notation $\log(n)$ to denote $\lceil \log_{2}(n)\rceil$ for $n\in\Bbb{Z}_{\geq 1}$
and we agree that  $\log(0)=1$. Given a non-negative integer $n$, we denote by $[n]$
the set of all positive integers no bigger than $n$.
Given a finite set $A$ and an integer $s\in\Bbb{Z}_{\geq 0}$, we denote by ${A\choose s}$   (resp. ${A\choose \leq s}$ )
the set of all subsets of $A$ with exactly (resp. at most) $s$ elements.
%
%
%Given a collection ${\cal F}$ of sets or graphs, we define $\cupall{\cal F}=\bigcup_{S\in{\cal F}}S$.
%%
 Given a
 function $f:A\rightarrow\Bbb{Z}_{\geq 0}$ we define $\sum f=\sum_{x\in A}f(x)$.
An {\em $r$-allocation}  of a set $S$ is an $r$-tuple ${\cal V}=(V_{1},\ldots,V_{r})$ of, possibly empty, sets that are pairwise disjoint and
whose union is the set $S$. We refer to the elements of ${\cal V}$ as the {\em parts} of ${\cal V}$
and we denote by ${\cal V}^{(i)}$ the $i$-th part of ${\cal V}$, i.e., ${\cal V}^{(i)}=V_{i}$.

\vspace{0mm}\subsection{\textbf{{Definitions about graphs}}}

In this paper, when giving the running time of an algorithm
of some problem whose instance involves a graph $G$, we
agree that  $n=|V(G)|$ and $m=|E(G)|$.

All graphs in this paper are loopless and they may have multiple edges. The only  exception to this agreement is in Subsection~\ref{subsec:hom2la} where we also  allow loops.
%
%For a vertex set $S \subseteq V(G)$, we define $N_G(S)$ as the set of vertices in $V(G) \setminus S$ with at least one neighbor in $S$, and $N_G[S] = N_G(S) \cup S$.
%
If $G$ is a
 graph and $X$, $Y$ are  two disjoint vertex subsets
of $V(G)$, we define $\delta_{G}(X,Y)$ as the set of edges with
one endpoint in $X$ and the other in $Y$.
Given a graph $G$,  denote by ${\cal C}(G)$  the collection of all connected components of ${ G}$.
%Given an $S\subseteq V(G)$, we denote by $G[S]$ the subgraph of $G$ induced
%by $S$ and we also denote $G^{+}[S]=G[S\cup N_{G}(S)]$.

%

\vspace{0mm}\subsection{The list allocation problem}
\label{thloca}

\vspace{0mm}
We define the problem \la\ as follows.

\begin{center}
\fbox{\begin{minipage}{13.5cm}
\noindent {\sc \allocation} (\la)\\
\noindent {\sl Input:} A tuple $I=(G,r,\lambda,\alpha)$
where $G$ is a graph, $r\in\Bbb{Z}_{\geq 1}$, $\lambda:V(G)\rightarrow 2^{[r]}$,  and
$\alpha: {[r]\choose 2}\rightarrow \Bbb{Z}_{\geq 0}$.\\
%\noindent {\sl Parameter:} $w=\sum\alpha$.
{\sl Output:} An $r$-allocation ${\cal V}$ of $V(G)$ such that
\begin{enumerate}\setlength\itemsep{.2em}
\item[~~~]  ~~{\bf 1.} $\forall \{i,j\}\in {[r]\choose 2}$, $|\delta_{G}({\cal V}^{(i)},{\cal V}^{(j)})|= \alpha(i,j)$ and
\item[~~~] ~~{\bf 2.} $\forall {v\in V(G)}, \forall i\in[r]$, if $v\in {\cal V}^{(i)}$ then $i\in\lambda(v)$,
\end{enumerate}
or a correct report that no such $r$-allocation exists.
\end{minipage}
}
\end{center}

\noindent For simplicity, in the above definition we write $\alpha(\{i,j\})$ as $\alpha(i,j)$ and we agree  that $\alpha(i,j)=\alpha(j,i)$. Also, given an instance $I$ of \la, we denote\footnote{Given a function $\tau: A\rightarrow \Bbb{Z}_{\geq 0}$, we
denote $\sum\tau=\sum_{x\in A}\tau(x)$.} $w(I) =\sum\alpha$. We will also
use $w$ instead of $w(I)$ when
it is clear what is the instance we are working with. We  assume that the multiplicity of
each edge in $G$ does not exceed $w$ as, if this happens, then reducing it
to $w$ creates an equivalent instance of the problem.

In the definition of \la\ each vertex $v$ of $G$ carries a {\em list} $\lambda(v)$
indicating the  parts where $v$ can be possibly allocated.
Moreover, $\alpha$  is a function assigning weights  to pairs of parts in ${\cal V}$.
The weights defined by $\alpha$
prescribe the precise number of crossing edges between distinct parts of ${\cal V}$.

Notice that \la\ is an {\sf NP}-hard problem by a simple reduction from the
{\sc Max Cut} problem, asking whether,
for an input graph $G$ and some $w\in\Bbb{Z}_{\geq 0}$, whether there is a partition $V_{1}$, $V_{2}$ of $V(G)$ such
that there are {\sl exactly\footnote{It is straightforward to see that the standard
reduction from {\sc Nae-3-Sat}  also works when the question of {\sc Max Cut} asks for exactly $w$ crossing edges instead of at least $w$ crossing edges.}} $w$ edges each with endpoints in both $V_{1}$ and $V_{2}$. Indeed, given an instance $I=(G,w)$ of {\sc Max Cut}, construct the instance $I'=(G,2,\lambda,
\alpha)$ where $\lambda(v)=\{1,2\}$ for every $v\in V(G)$ and $\alpha(1,2)=w$.
Note also that when $r=2$, \la\ is polynomially solvable on planar graphs as it directly
reduces to {\sc Planar Max Cut} that is polynomially solvable~\cite{Had75}.
%\hrule
%Moreover, \la\ remains {\sf NP}-hard on planar graphs for every fixed $r\geq 3$.
%For this, we present a simple reduction from the \textsc{Partition} problem. Given an instance $I=\{a_1, \ldots, a_t\}$ of \textsc{Partition}, we build an instance $I'=(G,3,\lambda,\alpha)$ of \la\ as follows: Let us start to construct a planar graph $G$ from an isolated vertex $r$ where $\lambda(r)=1$. For each $i \in [t]$, add to $G$ $a_i$ new vertices linked by a path, and connect each of these vertices to $r$.
%For each vertex $v$ different from $r$ we set $\lambda(v)=\{2,3\}$.
%Set $\alpha(2,3)=0$ and $\alpha(1,2)=\alpha(1,3)= \sum_{i=1}^t a_i/2$. It is clear that
%$I'$  is a  {\sc yes}-instance of \la\ if and only if $I$ is a {\sc yes}-instance of {\sc Partition}.
%Using the standard terminology of parametrized complexity,
%this reduction proves that {\sc LA} is {\sf para-NP}-complete
%when parameterized by $r$.

\section{Main reductions}
 \label{hegfnrymn}

In this section we formally define   {\sc Min-Max-Multiway Cut} and
 {\sc List Digraph Homomorphism} and we reduce them to {\sc List Allocation}.

\vspace{0mm}\subsection{Min-Max-Multiway Cut}
\label{subsec:minmax2la}

The  {\sc Min-Max-Multiway Cut} problem is formally defined as follows:
\begin{center}
\fbox{\begin{minipage}{13.5cm}
 {\sc  {\sc Min-Max-Multiway Cut}}\\
 {\sl Input:}  A tuple $I=(G,\ell,r,T)$ where $G$ is an undirected graph,
 $\ell,r\in\Bbb{Z}_{\geq 0}$, and $T\subseteq V(G)$ with $|T|=r$.\\
%
%\noindent {\sl Parameter:} $w\cdot r$.
%
{\sl Output:} A partition  $\{{\cal P}_{1},\ldots,{\cal P}_{r}\}$ of $V(G)$ such that for every ${i\in[r]},$ it holds that $|{\cal P}_{i}\cap T|=1$ and  $|\delta_{G}({\cal P}_{i},V(G)\setminus {\cal P}_{i})|\leq \ell$, or a correct report that no such partition exists.
\end{minipage}
}
\end{center}

%\ig{formally, in the above definition, the ``Question'' is not a question, but an ``Output''. This happens with all problems. For all of them, I would just keep ``Question'' and remove the part saying ``Or correctly report...''}
%

Similarly to  the case of \la, we assume that the multiplicity of each
edge in $G$ does not exceed $\ell$.

\begin{theorem}\label{thm:minmax2la}
If there is an algorithm that solves \la\ in $T(n,w(I))$ steps, then there exists an
algorithm that solves  {\sc Min-Max-Multiway Cut} in $2^{O(r\cdot \min\{\ell\cdot  \log r,r\cdot \log \ell\})}\cdot T(n,r\ell)$ steps.
\end{theorem}

\begin{proof}
Given an input $I=(G,\ell,r,T)$ of   {\sc Min-Max-Multiway Cut},
we fix (arbitrarily) a bijection $\mu: V(T)\rightarrow [r]$ and we
define $\lambda: V(G)\rightarrow 2^{[r]}$ such that
\[\lambda(x)=\left\{\begin{array}{lll}
& [r] & \mbox{if $x\in V(G)\setminus T$} \\
& \{\mu(x)\} & \mbox{if $x\in T.$}
\end{array}\right.\]
We now
consider the family ${\cal U}(I)$ of instances of \la\
containing one element $I'=(G,r,\lambda,\alpha)$
for each choice of function
 $\alpha: {[r]\choose 2}\rightarrow \Bbb{Z}_{\geq 0}$ satisfying
$$\forall i\in[r], \sum_{j\in[r]\setminus i}\alpha(i,j)\leq \ell.$$
Notice that $I$ is a {\sc yes}-instance of  {\sc Min-Max-Multiway Cut} if and only if
there exists some $I'\in{\cal U}(I)$ that is a {\sc yes}-instance of  \la. As  $|{\cal U}(I)|= 2^{O(r\cdot \min\{\ell\cdot  \log r,r\cdot \log \ell\})}$
% \sed{\small Reviewer says: Say why is $|U(I)|=2^{O(\ell\cdot  log r)}$}
 and for each $I'\in{\cal U}(I)$ it holds that $w(I')=O(r  \ell)$, the result follows.
\end{proof}

\vspace{0mm}\subsection{{List Digraph Homomorphism}}
\label{subsec:hom2la}

Let $G$ and $H$ be directed graphs where $G$ is simple  and $H$ may have loops but not multiple directed edges. A (directed) edge in the  digraph $G$ from the vertex $x$ to the vertex $y$ is denoted by $(x,y)$.
Let also $\lambda: V(G)\rightarrow 2^{V(H)}$.
 We denote by $E_{1}(H)$ the loops of $H$ and
by $E_{2}(H)$ the edges of $H$ between distinct vertices.
An {\em $\lambda$-{list} $H$-homomorphism} of $G$ is a function
$\chi:V(G)\rightarrow V(H)$ such that

\begin{itemize}
\item $\chi(v) \in \lambda(v)$ for every $v \in V(G)$, and
\item $(\chi(u),\chi(v))\in E(H)$ for every $(u,v)\in E(G)$.
\end{itemize}

Given a list  $H$-homomorphism $\chi$ of $G$ and an edge $e=(a,b)\in E_{2}(H)$
we define
%the {\em  $\chi${-arc charge}} of $e$ as
 %the
 %set
  $$C(e)=\{(u,v)\in E(G)\mid \chi(u)=a\mbox{ and } \chi(v)=b\}.$$
{\sc Bounded List Digraph Homomorphism} is formally defined as follows.

\begin{center}
\fbox{\begin{minipage}{13.5cm}
\noindent {\sc Bounded List Digraph Homomorphism} ({\sc BLDH})\\
\noindent {\sl Input:}\! A tuple $I=(G,H,\lambda,\ell)$
where $G$ and $H$ are  digraphs,  $\lambda:V(G)\rightarrow 2^{V(H)}$,\!   and\! $\ell\in\Bbb{N}_{\geq 0}$.\\
%\noindent {\sl Parameter:} $w=\sum\alpha$.
{\sl Output:} A {\em $\lambda$-{list} $H$-homomorphism} of $G$  where $ \sum_{e\in E(H)}|C(e)|\leq \ell$
or a correct report that no such homomorphism exists.
\end{minipage}
}
\end{center}
\medskip
%
%Given a function $\alpha: E_{2}(H)\rightarrow \Bbb{Z}_{\geq 0},$
%we say that $\chi$ is $\alpha${\em -arc-specified} if  $|C(e)|= \alpha(e),$
%for every $e\in E_{2}(H)$.\medskip
%
%
We now define the following more general problem.
\begin{center}
\fbox{\begin{minipage}{13.5cm}
 {\sc Arc-Specified List Digraph Homomorphism} ({\sc ASLDH})\\
 {\sl Input:}  A tuple $I=(G,H,\lambda,\alpha)$
where $G$ and $H$ are  digraphs,  $\lambda:V(G)\rightarrow 2^{V(H)}$,\!   and  $\alpha: E_{2}(H)\to \Bbb{Z}_{\geq 0}$.\\
\noindent{\sl Output:} A  $\lambda$-list $H$-homomorphism $\chi$
of $G$ such that  $\forall_{e\in E_{2}(H)}\ |C(e)|= \alpha(e)$ or a correct report that no such  $\lambda$-list $H$-homomorphism
 exists.
 \end{minipage}}
 \end{center}
Given an instance  $I=(G,H,\lambda,\alpha)$ of  {\sc ASLDH} we define $d(I)=\sum\alpha$.
As we already did for the cases of \la\ and  {\sc Min-Max-Multiway Cut}, we assume
that the multiplicity of the edges of the instance of {\sc BLDH} (resp. {\sc ASLDH}) does
not exceed $\ell$ (resp. $d(I)$).
%Using essentially the same reduction from {\sc Partition} as we did for the case of {\sc LA}
%it follows that {\sc ASLDH} is {\sf para-NP}-complete when parameterized by $d=d(I)$.

In the next sections we will prove that   there exists an {\sf FPT}-algorithm for {\sc ASLDH},  when parameterized by
{\sl both} $h=|V(H)|$ and $d=d(I)$. This fact together with the following result yields Theorem~\ref{ldh}.

\begin{theorem}
\label{kl6iq2zmn6}
If there is an algorithm that solves {\sc ASLDH} in $T(n,d(I))$ steps, then there exists an
algorithm that solves  {\sc BLDH}  in $2^{O(\ell\log h)}\cdot T(n,\ell)$ steps where $h=|V(H)|$.
\end{theorem}

\begin{proof}
%[Proof of Theorem~\ref{kl6iq2zmn6}]
Given an instance $I=(G,H,\lambda,\ell)$ of {\sc BLDH}  we set
${\cal U}(I)=\{(G,H,\lambda,\alpha)\mid \sum\alpha\leq \ell\}$
and we observe that $I$ is a {\sc yes}-instance of {\sc BLDH}  if and only if
some $I'\in{\cal  U}(I)$ is a {\sc yes}-instance of {\sc ASLDH}.
The lemma follows as $|{\cal U}(I)|=2^{O(\ell\log h)}$ and $d(I')\leq \ell$.
\end{proof}

\vspace{0mm}\subsection{A sparsifier for {\sc ASLDH}}
\label{sput8}

In order to prove that {\sc ASLDH} admits an {\sf FPT}-algorithm when parameterized by
{\sl both} $h=|V(H)|$ and $d=d(I)$,  we
will give a Turing-{\sf FPT} reduction of  {\sc ASLDH} to \la\ in Subsection~\ref{i7p1m6iowq}. The latter problem can be solved by an {\sf FPT}-algorithm due to the result of Section~\ref{sec:lafpt}.
The reduction of Subsection~\ref{i7p1m6iowq} receives an instance $(G,H,\lambda,\alpha)$ of    {\sc ASLDH}
 and returns an equivalent  instance $(G',r,\lambda',\alpha')$ of \la\ where
$|V(G')|=O(|E(G)|)$ which is $O(w\cdot |V(G)|^{2})$, in general. In order to avoid this blow-up in the polynomial
running time of our final {\sf FPT}-algorithm, we  give a way to transform the instances of   {\sc ASLDH}
to equivalent instances of the same problem  whose graphs are sparse.
This ``sparsification'' procedure  is described below.\medskip

A graph is {\em $d$-edge connected} if it has at least two vertices and for every
two vertices there are $d$ edge disjoint paths between them.
We use the following result from~\cite{KirousisSS93}.
\begin{proposition}
\label{skr3}
For every $d\in \Bbb{Z}_{\geq 1}$,
every graph $G$ where $|E(G)|\geq  d\cdot (|V(G)|-1)$ contains a $d$-edge connected subgraph.
\end{proposition}

We need first the following known result. For completeness, we  provide the proof.

\begin{lemma}
\label{op34klri}
Let $G$ be a $d$-edge connected graph and let $S=(s_{1},\ldots,s_{d})$ and $T=(t_{1},\ldots,t_{d})$
be two orderings of vertices of $V(G)$, possibly with repetitions.
Then, there exists a bijection $\sigma: [d]\rightarrow [d]$
and a collection ${\cal C}$ of $d$ pairwise edge-disjoint
paths such that
 for each $i\in[d]$, $s_{i}$ and $t_{\sigma(i)}$ are the endpoints of some  path in ${\cal C}$.
\end{lemma}

\begin{proof}
We add in $G$ two new vertices $s$ and $t$ and connect $s$ with each vertex in $S$
and $t$ with each vertex in $T$ such that the multiplicity of each edge $\{s,x\}$ is
equal to the number of  times $x$ appears in $S$ and multiplicity of each edge $\{t,x\}$ is
equal to the number of  times $x$ appears in $T$. We observe that there are $d$-edge-disjoint paths from $s$ to $t$. To see this, suppose that a set of fewer than $d$ edges in $E(G')$ disconnects $s$ and $t$. This means that removal of fewer than $d$ edges of $E(G)$ disconnects $s_i$ and $t_j$ for some $i,j\in [d]$, contradicting $d$-edge-connectivity of $G$. Hence, we can find  $d$ edge-disjoint paths between $s$ and $t$ by Menger's Theorem. Removing $s$ and $t$ from these paths yields $d$ edge-disjoint paths between $S$ and $T$ having the desired property.
\end{proof}

The following Lemma is based on Lemma~\ref{op34klri}.

\begin{lemma}
\label{sjg4rtt}
Let $G$ be a graph and let ${\cal C}=\{C_{1},\ldots,C_{r}\}$ be a collection of vertex disjoint connected subgraphs of $G$.
Let also $G'$ be the graph obtained if we contract in $G$ all edges in the graphs in ${\cal C}$.
If $G'$ is $d$-edge connected and each  graph in ${\cal C}$ is $d$-edge connected or a single vertex, then $G$ contains a subgraph that is $d$-edge connected.
\end{lemma}

\begin{proof}
Let $H$ be the subgraph of $G$ induced by the vertices in the graphs in ${\cal C}$.
Given a vertex $v\in V(H)$, we  denote by $C_{v}$ the graph in ${\cal C}$
that is either $v$ itself or is contracted in $G$ to create $v$ in $H$.
We prove that for every two vertices $s$ and $t$ in $V(H)$ there are $d$-edge-disjoint paths between them. This follows easily in the case where both $s$ and $t$ belong in the same $C_{v}\in{\cal C}$ because of the $d$-edge connectivity of $C_{v}$.
Assume now that $x\in C^{(s)}$ and $y\in C^{(t)}$  where $C^{(s)}$ and $C^{(t)}$ are different graphs in ${\cal C}$. Let also $v_{s}$ and $v_{t}$
be vertices of $G'$ such that $C_{v_{s}}=C^{(s)}$ and $C_{v_{t}}=C^{(t)}$.
 As $G'$ is $d$-edge connected, there is a collection ${\cal P}=\{P_{1},\ldots,P_{d}\}$ of  $d$ edge-disjoint
paths in $G'$ from $v_{s}$ to $v_{t}$. We direct all these paths from $v_s$ to $v_t$
and we set $W=\bigcup_{i\in[d]}P_{i}$.  Let $v\in W$ and let $E_{v}$
be the set of edges in $G'$ incident to $v$.
Notice that $E_{v}$ has a partition $\{E^{1}_{v},\ldots,E_{v}^{d}\}$ such that $E_{v}^{i}$ are the edges of $P_{i}$ that are incident to $v$.
Clearly, each $E_{v_s}^{i}$ has only one edge and the same holds for each $E_{v_t}^{i}$.
Moreover, each $E_{v}^{i}$ with $v\not\in\{v_s,v_t\}$ has cardinality
two. We enhance the notation of
the sets $E_{v}^{i}$ as follows: if $E_{v}^{i}=\{e\}$ and  $v=v_s$ then we write $E_{v}^{i}=({\sf s},e)$.
If $E_{v}^{i}=\{e\}$ and  $v=v_t$ then  we write $E_{v}^{i}=(e,{\sf t})$.
If $v\not\in\{v_s,v_t\}$ and $E_{v}^{i}=\{e,e'\}$ such that $e$ is ingoing to $v$ in $P_{i}$ and $e'$ is outgoing to $v$ in $P_{i}$ then we write $E_{v}^{i}=(e,e')$.
We now define the pair ${\bf p}_{i}^{v}$ as follows:
if $E_{v}^{i}=({\sf s},e)$ and $y$ is the endpoint the edge $e$ in $G'$ that belongs in $C_{s}$ then we set ${\bf p}_{i}^{v}=(s,y)$,
if $E_{v}^{i}=(e,{\sf t})$ and $y$ is the endpoint the edge $e$ in $G'$ that belongs in $C_{t}$, then we set ${\bf p}_{i}^{v}=(y,t)$, and
if $E_{v}^{i}=(e,e')$, then $y$ and $y'$ are the endpoint the edges $e$ and $e'$ respectively
that belong in $C_{v}$, then ${\bf p}_{i}^{v}=(y,y')$.
For each $v\in W$ we create two orderings $S_{v}=(s_{1}^{v},\ldots,s_{d}^{v})$ and $T_{v}=(t_{1}^{v},\ldots,t_{d}^{v})$ of vertices in $C_{v}$ such that $(s_{i}^{v},t_{i}^{v})={\bf p}_{i}^{v}$ for every $i\in[d]$. For each $v$, we apply Lemma~\ref{op34klri} and obtain a collection ${\cal P}_{v}$
of edge-disjoint paths between the vertices of $S_{v}$ and the vertices of $T_{v}$.
It is now easy to observe that the subgraph of $H$ consisting
of the edges in the paths in ${\cal P}$ (that are also edges of $G'$)
and the edges of the paths in ${\cal P}_{v}$ for every $v\in W$ is the union
of $d$ edge-disjoint paths in $H$ between $s$ and $t$.
\end{proof}

Given a graph $H$ and a positive integer $d$, we say
that a subgraph $H$ of $G$ is a {\em $d$-edge connected core of $G$}
if every connected component of $H$ is
$d$-edge connected and, among all such subgraphs of $G$, $H$ has maximum number of edges. The proof of the next lemma uses Proposition~\ref{skr3}.

\begin{lemma}
\label{skr4}
For every $d\in\Bbb{Z}_{>0}$, every graph $G$ with $m\geq  d\cdot (n-1)$
contains a unique $d$-edge connected core that can be found in $O(d\cdot n^4)$ steps.
\end{lemma}

%\removed{
\begin{proof}
The claimed $d$-edge connected core exists because of Proposition~\ref{skr3}.
Also, it is unique because if there are two $d$-edge connected cores $J_{1}$ and $J_{2}$, then it can be easily
checked that
the graph $J_{1}\cup J_{2}$ is also a $d$-edge connected core of $G$.
The algorithm
 repetitively removes from $G$ edges of min-cuts of size at most $d-1$ in its connected components  (each can be found in $O(d\cdot n^{3})$
steps according to~\cite{StWa97})
until this is not possible anymore (isolated vertices, when appearing during this procedure, are removed).

Note that the total number of steps of this procedure is bounded by the running time of the algorithm in~\cite{StWa97} times
the number of connected components of the resulting graph. This justifies the claimed running time.
Let $J$ be the  $d$-edge connected core  of ${G}$. Notice that none of the edges
of $J$ will be deleted by this procedure. Indeed, assuming the opposite,
let $G'$ be the graph where for the first time
a cut $(V_{1},V_{2})$ is found where the set $F$ of crossing edges contains some edge $e=\{x,y\}$ in $J$. Let also $C$ be the connected component
of $G'$ containing this cut and let $C_{J}$ be a  connected component of $J$ that is a subgraph of $C$
containing $e$.

Notice that $x$ and $y$ belong to different connected components of $C\setminus F$ and
therefore also to different connected components of $C_{J}\setminus F$, contradicting the fact that $C_{J}$
is $d$-edge connected. We just proved that the output of the algorithm will be a subgraph of $J$.
Notice also that each connected component of this output is  $d$-edge connected. By the maximality of $J$,
this output is necessarily $J$.
\end{proof}

\begin{lemma}
\label{sparsehom}
There is an $O(d(I)\cdot n^{4})$-step algorithm that given in instance $I=(G,H,\lambda,\alpha)$ of
{\sc ASLDH}, outputs an equivalent instance $I'=(G',H,\lambda',\alpha)$ of the same problem
where
$|E(G')|=O(d(I)\cdot |V(G')|)$.
\end{lemma}

\begin{proof}
Let $\tilde{G}$ be the
underlying graph of $G$ (multiplicities of edges of opposite  direction are summed up)
and $d=d(I)$.
If $\tilde{G}$ does not contains a  $(d+1)$-edge connected core, then, from
Proposition~\ref{skr3}, $|E(G)|=O(d\cdot |V(G)|)$.

Suppose now that $\tilde{G}$ has a $(d+1)$-edge connected core $J$ that, from Lemma~\ref{skr4},
can be found in~$O(d\cdot |V(G)|^{4})$ steps. We create a new graph $G'$
as follows: for each  $C\in {\cal C}(J)$ we contract all  vertices
of $C$ to a single vertex $v_{C}$ and we update $\lambda$ to $\lambda'$
so that if $x\not\in\{v_{C}\mid C\in {\cal C}(J)\}$, then $\lambda'(x)=\lambda(x)$
and  if $x=v_{C}$, then  $\lambda'(x)=\cap_{y\in V(C)}\lambda(y)$.
We claim that  $I'=(G',H,\lambda',\alpha)$ is an equivalent instance of {\sc ASLDH}. Indeed, this is based on the fact that, given a $\lambda$-list $H$-homomorphism $\chi$
of $G$ and a connected component $C$ of $J$, all vertices of $J$ should
be the preimages via $\chi$ of the same vertex of $H$. To verify this fact, just observe that, if this is not the case,
then the removal of the $\leq d$ crossing edges from $C$  (i.e., edges with endpoints mapped to different vertices of $H$)
will disconnect $C$, a contradiction to the $(d+1)$-edge-connectivity of $C$.

It now remains to prove that
$|E(G')|=O(d\cdot |V(G')|)$. If $|E(G')|\geq  (d+1)\cdot (|V(G)'|-1)$, then,  again from
Proposition~\ref{skr3},  $G'$ contains a $(d+1)$-edge connected subgraph. This, because of Lemma~\ref{sjg4rtt}, implies that  $G$ contains a subgraph
that is $(d+1)$-edge connected and has more edges than $J$, a contradiction.
\end{proof}

\vspace{0mm}\subsection{A reduction of  {\sc ASLDH} to \la}
\label{i7p1m6iowq}

Given the results of the previous section we are now in position to prove the following.

\begin{theorem}\label{thm:hom2la}
If there is an algorithm that solves \la\ in $T(n,w(I))$ steps, then there exists an
algorithm that solves  {\sc ASLDH}  in $T\big(O(d(I)\cdot n),O(d(I))\big)+O(d(I)\cdot n^{4})$ steps.
\end{theorem}

%\sed{\small Reviewer says:  In the full version I would suggest to separate the two directions
%in proof of Theorem 9 into two separate lemmas.}
\begin{proof}
Let $I=(G,H,\lambda,\alpha)$ be an instance of {\sc ASLDH}.
Using the algorithm of Lemma~\ref{sparsehom}, we may assume that $|E(G)|=O(d(I)\cdot |V(G)|)$.
We then use $I$ to  generate an instance $I'=(G',r,\lambda',\alpha')$ of \la, as follows:\medskip
\begin{itemize}%\setlength\itemsep{.2em}
\item[$\bullet$]$G'=(V',E')$, where\medskip
\begin{itemize}%\setlength\itemsep{.2em}
\item[$\circ$]$V'=V\cup V_{F}\cup V_{L}$, where
$V=V(G)$, $V_{F}=\{f_{uv}\mid (u,v)\in E(G)\},$ and $V_{L}=\{\ell_{uv}\mid (u,v)\in E(G)\}$ and\medskip
\item[$\circ$] $E'= E\cup E_{F}\cup E_{L}$, where
$E=\{\{f_{uv},\ell_{uv}\}\mid (u,v)\in E(G)\}$, $E_{F}=\{\{u,f_{uv}\}\mid (u,v)\in E(G)\}$, $E_{L}=\{\{\ell_{uv},v\}\mid  (u,v)\in E(G)\}$.
\end{itemize}\medskip
\item[$\bullet$]$r=|V(H)|+2\cdot |E_{2}(H)|$ and  $\sigma: V(\tilde{H})\rightarrow [r]$
is a   bijection where $\tilde{H}$ is the graph obtained from $H$ by subdividing twice each of its arcs that are not loops. For each arc  $(x,y)\in E_{2}(H)$, we denote its corresponding path in $\tilde{H}$ as $P_{xy},$ where $V(P_{xy})=\{x,\tilde{f}_{xy},\tilde{\ell}_{xy},y\}$.\medskip
\item[$\bullet$] $\lambda':V(G')\rightarrow [r]$ such that
\[\lambda'(w)=\left\{\begin{array}{lll}
& \{\sigma(x)\mid x\in \lambda(w)\} & \mbox{if $w\in  V$} \\ \\
&\{\sigma(\tilde{f}_{xy})\mid x\in \lambda(u)\wedge  y\in\lambda(v)\wedge  x\neq y \}
\ \cup& \\
&  \{\sigma(x)\mid x\in \lambda(u)\cap\lambda(v)\wedge (x,x)\in E_{1}(H) \}& \mbox{if $w=f_{uv}\in V_{F}$}\\ \\
&\{\sigma(\tilde{\ell}_{xy})\mid x\in \lambda(u)\wedge  y\in\lambda(v)\wedge  x\neq y \}
\ \cup& \\
&  \{\sigma(x)\mid x\in \lambda(u)\cap\lambda(v)\wedge (x,x)\in E_{1}(H) \}& \mbox{if $w=\ell_{uv}\in V_{L}$}.\\\end{array}\right.\]
%\item[$\bullet$]$w=\max \alpha$.
\item[$\bullet$]$\alpha': {[r] \choose 2}\rightarrow \Bbb{Z}_{\geq 0}$ such that
\[\alpha'(i,j)=\left\{\begin{array}{lll}
& \alpha(x,y) & \mbox{if  there exists some $(x,y)\in E_{2}(H)$ such that} \\
& &  \mbox{$(i,j)\in\big\{(\sigma(x),\sigma(\tilde{f}_{xy})),(\sigma(\tilde{f}_{xy})),\sigma(\tilde{\ell}_{xy})),(\sigma(\tilde{
\ell}_{xy}),\sigma(y))\big\}$} \\ \\
& 0 & \mbox{otherwise.}\\
\end{array}\right.\]
\end{itemize}

%{\tt [From here, written by Eunjung]}
Let $\chi:V(G)\rightarrow V(H)$ be a $\lambda$-list $H$-homomorphism of $G$ where  $\forall_{e\in E_{2}(H)}\ |C(e)|= \alpha(e)$. We construct an $r$-allocation ${\cal V}$ of $V(G')$ as follows:\medskip
\begin{itemize}\setlength\itemsep{.2em}
\item[$\bullet$]for every $u\in V=V(G)$, $u$ belongs to the part ${\cal V}^{i}$, where $i=\sigma(\chi(u))$\medskip
\item[$\bullet$]for every $f_{uv}\in V_F$, $f_{uv}$ belongs to the part ${\cal V}^{i}$, where \[i=\left\{\begin{array}{lll}
\sigma(\chi(u))& \text{if } \chi(u)=\chi(v)\\
\sigma(\tilde{f}_{xy})& \text{if } x=\chi(u)\neq y=\chi(v)
\end{array}\right.\]\medskip
\item[$\bullet$]for every $\ell_{uv}\in V_L$, $\ell_{uv}$ belongs to the part ${\cal V}^{i}$, where \[i=\left\{\begin{array}{lll}
\sigma(\chi(u))& \text{if } \chi(u)=\chi(v)\\
\sigma(\tilde{\ell}_{xy})& \text{if } x=\chi(u)\neq y=\chi(v)
\end{array}\right.\]
\end{itemize}
It is easy to verify that ${\cal V}$ is a solution for $I'$.\medskip

Now consider a solution ${\cal V}$ for $I'$. From ${\cal V}$, we define a mapping $\chi:V(G)\rightarrow V(H)$ so that for every $u\in V$, we have that $\chi(u)=\sigma^{-1}(i)  \text{~if and only if~}  u\in {\cal V}^{(i)}.$ We claim that $\chi$ is a $\lambda$-list $H$-homomorphism of $G$ where  $\forall_{e\in E_{2}(H)}\ |C(e)|= \alpha(e)$. For this, we investigate $\chi$ upon two conditions: firstly, we verify that $\chi$ is a $\lambda$-list $H$-homomorphism, and secondly that  $\forall_{e\in E_{2}(H)}\ |C(e)|= \alpha(e)$.\medskip

Let us prove that $\chi$ is a $\lambda$-list $H$-homomorphism. To see that $\chi(u)\in \lambda(u)$ for every $u\in V(G)$, let $u$ be in the $i$-th part of ${\cal V}$. Since $i \in \lambda'(u)$, the construction of $\lambda'$ implies that $\sigma^{-1}(i)\in \lambda(u)$, and thus $\chi(u)\in \lambda(u)$. To see that $\chi$ is an $H$-homomorphism, for an arbitrary edge $(u,v)\in E(G)$ we shall show that $(\chi(u),\chi(v))\in E_1(H)\cup E_2(H)$. Let $u$ and $v$ respectively belong to $\sigma(x)$-th and $\sigma(y)$-th parts of ${\cal V}$, for some $x,y\in V(\tilde{H})$. Note that $x\in \lambda(u)\subseteq V(H)$ and $y\in \lambda(v)\subseteq V(H)$. There are two possibilities: $x\neq y$ or  $x=y$.\medskip

\noindent {\em  Case 1:} $x\neq y$. Since $\sigma$ is a bijection, this means $\sigma(x)\neq \sigma(y)$. From the way we construct $\alpha'$, the vertices $f_{uv}$ and $\ell_{uv}$ can be only allocated into the $\sigma(\tilde{f}_{xy})$-th part and the $\sigma(\tilde{\ell}_{xy})$-th part, respectively, in the solution ${\cal V}$. Furthermore, the construction of $\alpha'$ also implies $(x,y)\in E_2(H)$.\medskip

\noindent {\em Case 2:} $x=y$. This means $\sigma(x)=\sigma(y)$. The construction of $\alpha'$ implies $f_{uv}$ and $\ell_{uv}$ are allocated into the $\sigma(x)$-th part of ${\cal V}$ as well. This, in turn, means that $\sigma(x)\in \lambda'(f_{uv})$ and $\sigma(x)\in \lambda'(\ell_{uv})$. Recall that $\lambda'(f_{uv})$ contains $\sigma(x)$ only when $(x,x)\in E_1(H)$. Hence, $(x,y)\in E_1(H)$.\medskip

Now we verify that  $\forall_{e\in E_{2}(H)}\ |C(e)|= \alpha(e)$. Consider an arc $e=(x,y)\in E_2(H)$. Note that for every directed edge $(u,v)$ in the $\chi$-arc charge $C(e)$, the $(u,f_{uv})$ of $E(G')$ contributes to $\alpha'(\sigma(x),\sigma(\tilde{f}_{xy}))$ exactly by one unit. Conversely, for every edge $(u,f_{uv})$ of $E(G')$ which contributes to $\alpha'(\sigma(x),\sigma(\tilde{f}_{xy}))$, we have $\chi(v)=y$ and thus the directed arc $(u,v)$  contributes to $C(e)$ by one unit. This establishes that  $\forall_{e\in E_{2}(H)}\ |C(e)|= \alpha(e)$.
\smallskip

The claimed running time follows from the fact that $w(I')=\sum\alpha'=3\cdot\sum\alpha =O(d(I))$
and $|V(G')|=O(|E(G)|)=O(d(I)\cdot |V(G)|)$.
\end{proof}

 \vspace{0mm}\section{An {\sf FPT}-algorithm for List Allocation}
\label{sec:lafpt}

In this section we give the proof that \la\ admits an {\sf FPT}-algorithm.
Before we proceed with the details of the proof let us summarize the main steps
of the proof that consists of a series of  {\sf T-FPT}-reductions.

\begin{enumerate}
\item  {\sc List Allocation} is {\sf T-FPT}-reduced to its restriction, called {\sc CLA}, where $G$ is a connected graph and only $O(w)$ boxes
are used. This reduction
takes care of the different ways connected components of $G$ can entirely be placed into the boxes
and is based on dynamic programming (see~Subsection~{\bf \ref{subsfs4r}}).\medskip

\item  {\sc CLA} is {\sf T-FPT}-reduced  to a restriction of it, called {\sc HCLA}, where $G$ is highly connected in the sense
that there is no set of $w$ edges that can separate $G$ into two ``big'' connected components.
This reduction  is presented in detail in Subsection~{\bf \ref{jj8poo4yyy}} and uses the technique of  {\em recursive
understanding}, introduced in~\cite{KawarabayashiT11them} and further developed in~\cite{CyganLPPS14minb} and~\cite{ChitnisCHPP12desi} (see also~\cite{GroheKMW11find}%
% \ig{are you sure that recursive understanding is also used in~\cite{GroheKMW11find}? sed> YES
%Don't you mean the paper about MIN BISECTION~\cite{CLPPS14}? This
%reference about MIN BISECTION was missing before!!}
), for generalizations of the  {\sc Multiway Cut} problem).\medskip

\item  {\sc HCLA} is {\sf T-FPT}-reduced  to a special enhancement of it, called \shcla, whose input additionally contains some
set $S\subseteq V(G)$  and the problem asks for a solution where all vertices of $S$ are placed in a unique ``big'' box
and all vertices of this box which are incident to crossing edges are contained in $S$. This variant
of the problem permits the application of the technique of  {\em randomized contractions},  introduced in~\cite{ChitnisCHPP12desi} (see~Subsection~{\bf \ref{mnop93e}}).\medskip

\item  Finally, \shcla\ is {\sf T-FPT}-reduced  to {\sc List Allocation} restricted to instances whose sizes
are bounded by a function of the parameter.  This
is presented in Subsection~{\bf \ref{mnr34op973e}}. It
is a  dynamic programming based on
the fact that an essentially equivalent instance of the problem
can be constructed if, apart from $S$, we remove from $G$ all but a bounded number
of the connected components of $G\setminus S$.
\end{enumerate}

\subsection{Some (more) definitions}
Given two sets $A$ and $B$ we denote by $B^A$ the set containing every function $f:A\rightarrow B$.
Given a function $h:A\rightarrow B$
and $S\subseteq A$,
 we define  $h|_S=\{(x,y)\in h\mid x\in S\}$.
 Given two functions $f_1,f_2:A\rightarrow\Bbb{Z}_{\geq 0}$ we define $f_{1}+f_{2}: A\rightarrow\Bbb{Z}_{\geq 0}$
such that $(f_{1}+f_{2})(x)=f_{1}(x)+f_{2}(x)$.
Let $X$ be a set and let $\zeta_{1},\zeta_{2}$ be two   functions mapping $X$ to non-negative integers. We say that $\zeta_{1}\leq \zeta_{2}$
if  $\forall {i\in X},\ \zeta_{1}(i)\leq\zeta_{2}(i)$.
Given a (possibly partial) function $\zeta: X\rightarrow\Bbb{Z}_{\geq 0}$
we define $\frak{F}_{\leq }(\zeta)=\{\zeta': X\rightarrow \Bbb{Z}_{\geq 0}\mid \zeta'\leq \zeta\}$.
Given a set $T\subseteq S$, we define the {\em restriction of ${\cal V}$
to $S$} as the $r$-allocation  ${\cal V}\cap {T}=({\cal V}^{(1)}\cap T,\ldots,{\cal V}^{(r)}\cap T)$. Notice that  ${\cal V}\cap {T}$ is an $r$-allocation of $T$. Given two $r$-allocations ${\cal V}_{1}=(V_{1}^{1},\ldots,V_{r}^{1})$ and ${\cal V}_{2}=(V_{1}^{2},\ldots,V_{r}^{2})$, we
define ${\cal V}_{1}\cup {\cal V}_{2}=(V_{1}^{1}\cup V_{1}^{2},\ldots,V_{r}^{1}\cup V_{r}^{2})$.

Given two graphs $G$ and $G'$ we set $G\cup G'=(V(G)\cup V(G'),E(G)\cup E(G'))$.
Given a graph $G$ and a set $S\subseteq V(G)$,
we define $\partial_G(S)$ as the set of all vertices in $S$
that are adjacent to vertices in $V(G)\setminus S$.

Let $G$ be a connected graph. A partition $(V_{1},V_{2})$ of $V(G)$ is a {\em $(q,y)$-good
separation} if $|V_{1}|,|V_{2}|>q$,
$|\delta_{G}(V_{1},V_{2})|\leq y$, and $G[V_1]$ and $G[V_{2}]$ are both connected.
A graph $G$ is called {\em $(q,y)$-{connected}} if it does  not contain any
{\em $(q,y-1)$-good separation}. (Note that for $q=0$, $(q,y)$-connectivity corresponds exactly to classical $y$-edge-connectivity.)

\begin{proposition}[Chitnis {et al}.~\cite{ChitnisCHPP12desi}]
\label{bhtefdjfh}
There exists a deterministic algorithm that, with input a $n$-vertex connected graph $G$,
a $q\in \Bbb{Z}_{\geq 1}$ and $y\in\Bbb{Z}_{\geq 0}$, either finds a $(q,y)$-good separation,
or reports  that no such separation exists, in  $2^{O(\min\{q,y\}\cdot  \log(q+y))}n^3 \log n$ steps.
\end{proposition}

For a solution ${\cal V}$ to an instance $I=(G,r,\lambda,\alpha)$ of \la, we define $$E({\cal V})=\bigcup_{\{i,j\}\in{[r]\choose 2}}\delta_{G}({\cal V}^{(i)},{\cal V}^{(j)}).$$ We say that $H$ is $(i,\lambda)$-{friendly} if $i\in \bigcap_{v\in V(H)}\lambda(v)$.

\begin{observation}
\label{glojolg}
If ${\cal V}$ is a solution for some instance  $I=(G,r,\lambda,\alpha)$ of \la\  then every connected component of $G\setminus E({\cal V})$
is also a connected component of $G[{\cal V}^{(i)}]$ for some $i\in[r]$.
\end{observation}

\begin{observation}
\label{td7uyi}
If  ${\cal V}$ is a solution  for an instance $I=(G,r,\lambda,\alpha)$ of \la, where $G$ has $\ell$ connected components,
then $G\setminus E({\cal V})$ contains at most $w+\ell$ connected components.
\end{observation}

%
%\begin{observation}
%\label{kd7ehd8r}
%Each instance $I$ of \la\ has at most $w$ positive pairs and $2w$ positive indices.
%\end{observation}

\begin{lemma}
\label{j4oqndua}
There exists an algorithm that, given an instance $I=(G,r,\lambda,\alpha)$ of \la,
correctly solves the problem in $n^{O(w)}\cdot 2^{O((w+\ell)\cdot \log r)}$ steps, where $\ell$ is the
 number of connected components of $G$.
\end{lemma}

\begin{proof}
The algorithm considers each subset $F$ of $E(G)$ of size $w$. Notice that there are $n^{O(w)}$ such subsets.
From Observation~\ref{td7uyi},  $G\setminus F$
has at most $w+\ell$
connected components. From Observation~\ref{glojolg},
if ${\cal V}$ is a solution of \la\ for $I$,
and $E({\cal V})=F$, then the vertex set of each connected component of $G\setminus F$
is entirely contained in some ${\cal V}^{(i)}$.
The algorithm  considers all possible ways to assign the $\leq w+\ell$ connected components of $G_{F}$ to the  $r$ indices of $I$
 and checks whether this creates a solution for $I$. As there are $2^{O((w+\ell)\cdot \log r)}$ such assignments, the claimed running time follows.
\end{proof}

\vspace{0mm}\subsection{Connected list allocation}
\label{subsfs4r}

We define the {\sc Connected \allocation} problem ({\sc CLA}, in short) as the \allocation\ with the additional demand that the input graph $G$ is connected and $r\leq 2w$. The reason why we may assume that $r\leq 2w$ is the following. First, we may assume that $w\geq 1$ since otherwise, we can check whether $G$ is $(i,\lambda)$-friendly for some $i\in [r]$ and solve the instance for \cla. Suppose $r>2w>0$. Then there exists an index $i\in [r]$ such that $\alpha(i,j)=0$ for every $j\in [r]\setminus \{i\}$. As $G$ is connected, no vertex can be allocated to ${\cal V}^{(i)}$ in any solution ${\cal V}$ and thus we can remove the $i$-th part from the instance.
%

%Given an instance $I=(G,r,\lambda,\alpha)$
%of \cla,
% we define
% $${\bf folio}(I)=\{\alpha'\in\frak{F}_{\leq}(\alpha)\mid \mbox{$I=(G,r,\lambda,\alpha')$ is a {\sc yes}-instance of \cla}\}.$$
%
%Let ${\cal W}$ be a collection of connected subgraphs of $G$ and let $\lambda: V(G)\rightarrow [r]$.
%For each $\alpha'\in\frak{F}_{\leq}(\alpha)$ we define
%$${\bf rep}({\cal W},\alpha')=\{C\in{\cal W}\mid \alpha'\in {\bf folio}(C,r,\lambda|_{V(C)},\alpha) \}$$
%and, for every set $S\subseteq V(G)$,
%we define $${\bf trunk}(I,{\cal W},S)=\bigcup_{\alpha'\in \frak{F}_{\leq}(\alpha)}\tilde{\bf rep}({\cal W},\alpha')$$ where,  for each $\alpha'\in\frak{F}_{\leq}(\alpha)$, $\tilde{\bf rep}({\cal W},\alpha')$
%consists of $\min\{w,|{\bf rep}({\cal W},\alpha')|\}$  smallest, with respect to the number of vertices not in $S$, elements of ${\bf rep}({\cal W},\alpha')$.

%
%\begin{observation}
%\label{bgbg3ertu6}
%Given an instance $I$ of \cla, and a  collection ${\cal W}$ of connected subgraphs of $G$, and a set $S\subseteq V(G)$,
%the set ${\bf trunk}(I,{\cal W},S)$ contains at most $w\cdot 2^{w}$ elements and can be computed
%in $O(2^{O(w)}\cdot n)$ steps. Moreover, for every $\alpha'\in\frak{F}_{
%\leq}(\alpha^*)$, $|{\bf trunk}(I,{\cal W},S)\cap {\bf rep}({\cal W},\alpha')|\geq \min \{w,|{\bf rep}({\cal W},\alpha')\}$.
%\end{observation}

%
%The proof of the next lemma uses Observation~\ref{kd7ehd8r}.

\begin{lemma}
\label{d8h4kjks}
If there exists an algorithm solving {\sc CLA}
in $f(w)\cdot p(n)$ steps, then there is an algorithm that solves
\la\ in $\ell\cdot 2^{2w}\cdot f(w)\cdot p(n)$ steps.
\end{lemma}

\begin{proof}
%\todo{Use DP suggested by wise reviewer master!}
We present a dynamic programming for \la\ using an algorithm for {\sc CLA} as a subroutine. Let $C_1,\ldots , C_{\ell}$ be the connected components of $G$ and let $G_i=\bigcup_{j=1}^i C_j$. Define a table $P$ for dynamic programming in which the entries $P(i,\alpha')$ run over all $1\leq i\leq \ell$ and $\alpha'\in\frak{F}_{\leq}(\alpha)$. The value of $P(i,\alpha')$ is \yes\ if the instance $(G_i,r,\lambda|_{V(G_i)},\alpha')$ is \yes. Otherwise, $P(i,\alpha')=\no$. Note that the given instance $(G,r,\lambda,\alpha)$ is \yes\ to \la\ if and only if $P(\ell,\alpha)=\yes$.

For $i=1$, $G_1=C_1$ is connected and thus the value of $P(1,\alpha')$ can be correctly determined by solving {\sc CLA} on the instance $(C_1,r,\lambda|_{V(C_1)},\alpha')$. For $2\leq i\leq \ell$, we assume that all values $P(j,\alpha'')$ have been determined for $j<i$ and $\alpha'' \in\frak{F}_{\leq}(\alpha')$. Let $g$ be a function mapping instances of {\sc CLA} to $\{\yes,\no\}$ in a canonical way. The following recursion for $P(i,\alpha')$ is easy to verify.
\[P(i,\alpha')=\bigvee_{\alpha'' \in\frak{F}_{\leq}(\alpha')} P(i-1,\alpha'-\alpha'')\wedge g(C_i,r,\lambda|_{V(C_i)},\alpha'').\]

As $w=\sum \alpha$ by definition, the table $P$ consists of $\ell\cdot |\frak{F}_{\leq}(\alpha)|\leq \ell\cdot 2^w$ entries. Determining each entry amounts to at most $|\frak{F}_{\leq}(\alpha)|\leq 2^w$ table lookups and computations of $g$. The latter, equivalent to solving an instance to {\sc CLA}, takes at most $f(w)\cdot p(n)$ steps. Overall, the entire entries of $P$ can be determined in $\ell \cdot 2^{2w}\cdot f(w)\cdot p(n)$ steps.
 \end{proof}

%\vspace{0mm}
\vspace{0mm}\subsection{Highly connected list allocation}
\label{jj8poo4yyy}

We fix two functions $f_{1}(w)=2^{w}\cdot (2w)^{2w}$ and $f_2(w)= w \cdot f_1(w)+1$, which will appear through this section.
We define the {\sc Highly Connected \allocation} problem ({\hcla}, in short) as the {\sc Connected} \allocation\ problem with the only difference that we additionally  demand that the input graph is $(f_2(w),w+1)$-connected, where $w$ is the parameter of the problem.

We aim to shrink the size of a given instance $I=(G,r,\lambda,\alpha)$ of \cla\ by finding out a set $E_{{\cal C}}$ of edges such that $I/E_{{\cal C}}$ (formal definition given below) is equivalent to $I$. If it is possible to recursively contract edges so that the obtained instance is of size bounded by a function of $w$, then we can apply the algorithm of Lemma~\ref{j4oqndua} to solve the final instance of \cla.
\medskip

For an instance  $I=(G,r,\lambda,\alpha)$ of \cla\
and $B\in  {V(G)\choose \leq 2w}$,   we  set $\frak{U}(I,B)=[r]^{B}\times \frak{F}_{\leq}(\alpha)$.
Given a ${\bf w}=(\psi,\alpha')\in\frak{U}(I,B)$, we define the instance $I_{\bf w}=(G,\lambda',r,\alpha')$ of  \cla,
where $\lambda'=\lambda|_{V(G)\setminus B}\cup \psi$. The set $\tilde{\frak{U}}(I,B)$
is a collection of all ${\bf w}\in {\frak{U}}(I,B)$ such that $I_{\bf w}$ is a {\sc yes}-instance of  \cla.

\begin{observation}
\label{rj84hlsjd}
For every instance  $I=(G,r,\lambda,\alpha)$ of \cla\
and $B\in {V(G)\choose \leq 2w}$, it holds that $|\frak{U}(I,B)|\leq f_{1}(w) $.
\end{observation}
%
%We also set up the function
% $f_{2}:\Bbb{Z}_{\geq 0}\rightarrow\Bbb{Z}_{\geq 0}$ such that
%$f_2(w)=w\cdot (f_{1}(w))^{2}+2\cdot w+2.$

%Let $I=(G,r,\lambda,\alpha)$ be an instance of  \hcla\
%where $|V(G)|>f_{2}(w)$
%  and  let $B$ be a subset of $V(G)$ of at most $2\cdot w$ vertices.

Given an instance $I=(G,r,\lambda,\alpha)$  of \la\ and a set of edges $E_{{\cal C}} \subseteq E(G)$, we define the instance $I / E_{{\cal C}}$ of \la\ as $(G/E_{{\cal C}},r,\lambda',\alpha)$, where $G/E_{{\cal C}}$ is the graph obtained from $G$ by contracting all the edges in $E_{{\cal C}}$, and $\lambda'$ is defined as follows: for each vertex $u \in V(G/E_{{\cal C}})$, let $V_u \subseteq V(G)$ be the set of vertices of $G$ that have been identified into $u$ after contracting the edges in $E_{{\cal C}}$ (note that, possibly, $V_u = \{u\}$). Then we define $\lambda'(u) := \bigcap_{v \in V_u} \lambda(v)$.
%For a vertex set $B\subseteq V(G)$, the set of vertices obtained from $B$ by contracting edges in $E_{{\cal C}}$ is denoted as $B/E_{{\cal C}}$.

Let  $I=(G,r,\lambda,\alpha)$ be an instance of \cla\ and let $Q\subseteq V(G)$. We set $I[Q]=(G[Q],r,\lambda|_{Q},\alpha).$ For a bipartition $(V_1,V_2)$ of $V(G)$, let $E_{{\cal C}}$ be a set of edges in $G[V_1]$. Then the gluing of $G[V_1]/E_{{\cal C}}$ and $G[V_2]$ along $\delta_G(V_1,V_2)$, denoted as $G[V_1]/E_{{\cal C}}\oplus_{\delta} G[V_2]$, can be naturally defined: starting from the disjoint union of $G[V_1]/E_{{\cal C}}$ and $G[V_2]$, for each edge $e=(u,v)\in \delta_G(V_1,V_2)$ with $u\in V_1$, we add an edge $e'$ whose one endpoint is the vertex into which $u$ is identified and the other endpoint is $v$. Notice that $G[V_1]/E_{{\cal C}}\oplus_{\delta} G[V_2] = G/E_{{\cal C}}$.

\medskip

The next lemma demonstrates the condition for a set of edges $E_{{\cal C}}$ under which $I$ and $I / E_{{\cal C}}$ are equivalent.

\begin{lemma}\label{lem:contraction-equivalent}
Let $I=(G,r,\lambda,\alpha)$ be an instance of \cla, let $(V_{1},V_{2})$ be a bipartition  of $V(G)$
such that $I[V_{1}]$ is an instance of \hcla\ with $|V_1| > f_{2}(w)$, let $B \subseteq V_1$ with $|B| \leq 2w$, let ${\cal S}=\{{\cal V}_{\bf w}\mid {\bf w}\in \tilde{\frak{U}}(I[V_{1}],B)\}$ be a collection of (arbitrary chosen) solutions to $I[V_1]_{\bf w}$ for every ${\bf w}\in \tilde{\frak{U}}(I[V_{1}],B)$, and let $E_{{\cal C}} = E(G[V_1]) \setminus \bigcup_{{\cal V}_{\bf w} \in {\cal S}}E({\cal V}_{\bf w})$.
%let $E_{{\cal S}} = \bigcup_{{\cal V}_{\bf w} \in {\cal S}}E({\cal V}_{\bf w})$, and let $E_{{\cal C}} = E(G[V_1]) \setminus E_{{\cal S}}$.
Then $E_{{\cal C}}\neq \emptyset$. Furthermore, $I$ and $I / E_{{\cal C}}$ are equivalent instances of \cla.
\end{lemma}

\begin{proof}
%To see the first statement, note that $|\tilde{\frak{U}}(I[V_{1}],B)|\leq |\frak{U}(I[V_{1}],B)|\leq f_1(w)$, see Observation~\ref{rj84hlsjd}.  For each ${\cal V}_{\bf w}\in {\cal S}$, we have $|E({\cal V}_{\bf w})|=w$ and thus $|\bigcup_{{\cal V}_{\bf w} \in {\cal S}}E({\cal V}_{\bf w})|\leq w\cdot f_1(w) <f_2(w)$. As $G[V_1]$ is connected, $|E(G[V_1])|\geq f_2(w)$ and $E_{{\cal C}}\neq \emptyset$.

Note also that by Observation~\ref{rj84hlsjd}, $|\bigcup_{{\cal V}_{\bf w} \in {\cal S}}E({\cal V}_{\bf w})| \leq  w \cdot |\tilde{\frak{U}}(I[V_1],B)| \leq w \cdot f_1(w)$. Since by hypothesis the graph $G[V_1]$ is connected and satisfies $|V_1| > f_2(w)$, it holds that $|E(G[V_1])| \geq f_2(w)$, and thus $|E_{{\cal C}}| = |E(G[V_1]) \setminus \bigcup_{{\cal V}_{\bf w} \in {\cal S}}E({\cal V}_{\bf w})| \geq f_2(w) - w \cdot f_1(w)= 1$, hence there exists at least one edge in $E_{{\cal C}}$.

We need to prove that $I$ is a YES-instance of  \cla\ if and only if $I / E_{{\cal C}}$ is. First note that contracting edges does not harm the connectivity of $G$ and, as $r$ and $\alpha$ are the same in $I$ and in $I / E_{{\cal C}}$,  $r\leq 2\sum \alpha$ holds.  Therefore $I / E_{{\cal C}}$ is indeed an instance of \cla.

Assume first that $I$ is a YES-instance, and let ${\cal V}$ be a solution of \cla\ for $I$. Let $\psi_B = \{(v, {\cal V}(v)) \mid v \in B\}$, where ${\cal V}(v)$ denotes the integer $i \in [r]$ such that $v\in {\cal V}^{(i)}$. Let also $\alpha_1$ be the element of $\frak{F}_{\leq}(\alpha)$ such that for any two distinct integers $i,j \in [r]$, $\alpha_1(i,j) = |\delta_{G[V_1]}({\cal V}^{(i)},{\cal V}^{(j)})|$, and let ${\bf w}_1 = (\psi_B, \alpha_1)$. Note that by the definition of ${\bf w}_1$ and since we assume that $I[V_{1}]$ is an instance of \hcla, it holds that
${\bf w}_1\in\tilde{\frak{U}}(I[V_1],B)$. Let ${\cal V}_{{\bf w}_1}$ be the solution to $I[V_1]_{{\bf w}_1}$ in the collection ${\cal S}$ and note that by the definition of the set $E_{{\cal C}}$, the endpoints of any edge in $E_{{\cal C}}$ belong to the same part of ${\cal V}_{{\bf w}_1}$. We now proceed to define an $r$-allocation ${\cal V}'$ for $I / E_{{\cal C}}$.  For each vertex $u \in V(G/E_{{\cal C}})$, let $V_u \subseteq V(G)$ be the set of vertices of $G$ that have been identified into $u$ after contracting the edges in $E_{{\cal C}}$, so each of the sets $V_u$ belongs entirely to the same part of ${\cal V}_{{\bf w}_1}$. For every $u \in V_2$, we define ${\cal V}'(u)={\cal V}(u)$, and for every $u \in V(G/E_{{\cal C}}) \setminus V_2$, we define  ${\cal V}'(u)={\cal V}_{{\bf w}_1}(v)$, for an arbitrary vertex $v \in V_u$. The above discussion implies that the $r$-allocation ${\cal V}'$ is well-defined and constitutes a solution of \cla\ for $I / E_{{\cal C}}$.

Conversely, assume now that $I / E_{{\cal C}}$ is a YES-instance, and let ${\cal V}'$ be a solution of \cla\ for $I'$. For  each vertex $u \in V(G/E_{{\cal C}})$, let again $V_u \subseteq V(G)$ be the set of vertices of $G$ that have been identified into $u$ after contracting the edges in $E_{{\cal C}}$. We proceed to define an $r$-allocation ${\cal V}$ for $I$.  For every $v \in V(G)$, if $v \in V_u$ for a vertex $u \in V(G/E_{{\cal C}})$, we define ${\cal V}(v)={\cal V}'(u)$. We claim that the $r$-allocation ${\cal V}$ is a solution of  \cla\ for $I$. Indeed, by definition of $\lambda'$ of the instance $I / E_{{\cal C}}$, we have that for every vertex $v \in V(G)$ belonging to a set $V_u \subseteq V(G)$, it holds that ${\cal V}(v) \in \lambda'(u) = \bigcap_{w \in V_u} \lambda(w) \subseteq \lambda(v)$. On the other hand, since edge multiplicities are summed up when contracting edges, it holds that for any two distinct integers $i,j \in [r]$, $|\delta_{G}({\cal V}^{(i)},{\cal V}^{(j)})| = |\delta_{G / E_{{\cal C}}}({\cal V}'^{(i)},{\cal V}'^{(j)})| = \alpha(i,j)$.\end{proof}

Due to Lemma~\ref{lem:contraction-equivalent}, an edge set $E_{{\cal C}}$ to be contracted can be obtained if we can compute a collection of solutions ${\cal S}=\{{\cal V}_{\bf w}\mid {\bf w}\in \tilde{\frak{U}}(I[V_{1}],B)\}$ to the instances $I_{\bf w}$ of \hcla. This can be done by solving the instance $I_{\bf w}$ for each ${\bf w}\in \frak{U}([V_{1}],B)$ and this requires at most $f_1(w)$ iterations of \hcla-solver by Observation~\ref{rj84hlsjd}. This point is formalized in the next observation.

\begin{observation}
\label{jf7sgekf73hd}
If there exists  an algorithm that can find, if it exists, a solution of {\sc HCLA}
in $f(w)\cdot p(n)$ steps, then there is an algorithm that,
given an instance $I=(G,r,\lambda,\alpha)$ of \hcla\ and
a set $B\subseteq V(G)$ where $|B|\leq 2\cdot w$,
computes the set $\tilde{\frak{U}}(I,B)$,  a set  ${\cal S}=\{{\cal V}_{\bf w}\mid {\bf w}\in\tilde{\frak{U}}(I,B)\}$ in case $\tilde{\frak{U}}(I,B)\neq \emptyset$, and the set  $E_{{\cal C}}=E(G[V_1]) \setminus \bigcup_{{\cal V}_{\bf w} \in {\cal S}}E({\cal V}_{\bf w})$  in $f(w)\cdot p(n)\cdot f_{1}(w)$ steps.
\end{observation}

\IncMargin{2.5em}
\begin{algorithm}[h]
\LinesNumbered
\SetNlSty{texttt}{\bf    }{.     \ \  }
  \SetKwData{Let}{\bf let}
  \SetKwFunction{Return}{\bf return}
  \SetKwInOut{Name}{\bf Algorithm}
  \SetKwInOut{Input}{\sl Input}
  \SetKwInOut{Output}{\sl Output}
  \SetKwInOut{Parameter}{\sl Global Variable}
  \Name{{\bf shrink}($H,B$)}

  \Input{A graph $H$ and a set $B\subseteq V(H)$ s.t. $|B|\leq 2\cdot w$, $|V(H)|> f_{2}(w)$.}
  \Output{A graph $H^{\rm new}$ having at most $f_{2}(w)$ vertices or a report that $I$ is a {\sc no}-instance.}
  \Parameter{An instance $I'$ of  \cla. }	
  \BlankLine
  \eIf{$H$ has a $(f_{2}(w),w)$-separation $(V_{1},V_{2})$\label{k6tud8whh4}}{
 \Let $i$ be an integer in $\{1,2\}$ such that $|B\cap V_{i}|\leq w$ \\
% \Let $J={\delta}(V_{1},V_{2})$, $A_{i}=V_{i}\cap V(J)$, and $A_{3-i}=V_{3-i}\cap V(J)$ \\
\Let $B'= (B\cap V_{i}) \cup (V({\delta}(V_{1},V_{2})) \cap V_{i})$  \\
\Let $H'= {\bf shrink}(H[V_i],B')$, $H^{\rm new}=H' \oplus_{\delta} H[V_{3-i}]$ \label{23vFDgdf}\\
%\Let $I^{\rm new}=I^{\prime \rm new}\oplus_{\bf q} I'',$  where ${\bf q}=(A_{1},A_{2},J)$\\
\Let $B^{\rm new}$ be the vertices of $H^{\rm new}$ onto which the vertices of $B$ are identified\\
\If{$|V(H^{\rm new})|> f_2(w)$}{
{\bf return} ${\bf shrink}(H^{\rm new},B^{\rm new})$ \label{4tgergd} \\%
}
%{{\bf return} $\emptyset$ \\}
{\bf return} $H^{\rm new}$ \label{wneklrw}
}{
compute $\tilde{\frak{U}}(I'[V(H)],B)$ \label{5ufgdudgf}\\
\eIf{$\tilde{\frak{U}}(I'[V(H)],B)=\emptyset$}{report that $I$ is a {\sc no}-instance \\}{compute $ {\cal S}=\{{\cal V}_{\bf w}\mid {\bf w}\in\tilde{\frak{U}}(I'[V(H)],B)\}$, $E_{{\cal C}}=E(H)\setminus  \bigcup_{{\cal V}_{\bf w} \in {\cal S}}E({\cal V}_{\bf w})$ \label{5ufgdudgf} \\
\Let $I' \leftarrow I'/E_{{\cal C}}$ \label{43trgdfg}\\
{\bf return} $H/E_{{\cal C}}$ \label{23rwef}
}}
\end{algorithm}
\DecMargin{1em}

\cla\ is solved by reducing the instance size via iteratively contracting edges, and finding contractible edges boils down to solving \hcla\ by Observation~\ref{jf7sgekf73hd}. The next lemma formalize this as a  {\sf TFT}-reduction from \cla\ to \hcla. Although the idea of the reduction itself is straightforward, we provide rather a nontrivial reduction to achieve a better running time for \cla.

\begin{lemma}
\label{u83nd01nfbsyr}
If  {\sc HCLA} can be solved
in $f(w)\cdot p(n)$ steps,  then   {\sc CLA}  can be solved  in
$\max\{2^{O(w^2\cdot \log w)}\cdot n^{4}\cdot \log n,f(w)\cdot p(n)\cdot 2^{O(w\cdot \log w)}\}$ steps.
\end{lemma}
\begin{proof}
Let $I=(G,r,\lambda,\alpha)$ be an instance of  {\sc CLA}.
If $G$ has less than $f_{2}(w)$ vertices then, because of Lemma~\ref{j4oqndua} and the fact  that $r\leq 2w$, the problem can be solved in $2^{O(w^2\cdot \log w)}$ steps. If not, we call the algorithm {\bf shrink}$(G,\emptyset)$ with the given instance $I=(G,r,\lambda,\alpha)$ as the global variable $I'$. The global variable $I'$ is
initialized (only once) at the initial call {\bf shrink}$(G,\emptyset)$, and is considered to be out of the scope of the subsequent calls {\bf shrink}$(H,B)$. Hence, the subsequent calls share the access to the same global variable at line~\ref{43trgdfg}.

Let $G(I')$ refer to the input graph of the current global variable $I'$. We need the following claim, which ascertains the property of {\bf shrink}$(G,\emptyset)$ and also ensure that $I'[V(H)]$ referred to at line~\ref{5ufgdudgf} is a valid instance.

\begin{claim}\label{claim:subgraph+size}
The follow statements hold.
\begin{itemize}
\item[(a)] Throughout the execution of {\bf shrink}$(G,\emptyset)$, whenever a call {\bf shrink}$(H,B)$ is made, the graph $H$ is an induced subgraph of $G(I')$ for the current global variable $I'$.
\item[(b)] The graph\footnote{For notational convenience, we abuse {\bf shrink}$(H,B)$ also to denote the graph returned by the call {\bf shrink}$(H,B)$ whenever it is clear from the context.} returned by each call {\bf shrink}$(H,B)$ can be obtained by contracting edges of $H$ and has at most $f_2(w)$ vertices.
\end{itemize}
\end{claim}
\begin{proofofclaim}
Let $m$ be the number of calls for {\bf shrink}$(H,B)$ during the performance of {\bf shrink}$(G,\emptyset)$, including the initial call itself. At any step in {\bf shrink}$(G,\emptyset)$, let $i$ be the number of calls invoked so far (in Lines {\bf \ref{23vFDgdf}}
and~{\bf \ref{4tgergd}}) and $j$ be the number of return calls  (in Lines {\bf \ref{4tgergd}}, {\bf \ref{wneklrw}}, and~{\bf \ref{23rwef}}). Intuitively, $(i,j)$ stands for the current position in the recursion tree during the course of the algorithm {\bf shrink}$(G,\emptyset)$. Before we make the first call {\bf shrink}$(G,\emptyset)$, we have $(i,j)=(0,0)$. Clearly, $j\leq i$ during the entire execution, and the inequality is strict unless {\bf shrink}$(G,\emptyset)$ terminates. We prove the statements (a) and (b) by induction on $i+j$. Notice that the algorithm traverses from $(i,j)$ to $(i+1,j)$ exactly when a $(i+1)$-st call {\bf shrink}$(H,B)$ is made, and it traverses from $(i,j)$ to $(i,j+1)$ when a call {\bf shrink}$(H,B)$ returns a $(j+1)$-st output. Therefore, we prove the following modification of (a) and (b).

\begin{itemize}
\item[]($a'$) Whenever the traversal $(i,j)\rightarrow (i+1,j)$ is made by a call {\bf shrink}$(H,B)$, $H$ is an induced subgraph of $G(I')$, where $I'$ is the current global variable. \medskip
\item[]($b'$) Whenever the traversal $(i,j)\rightarrow (i,j+1)$ is made by a return of {\bf shrink}$(H,$
$B)$, the returned graph can be obtained by contracting edges of $H$, and has at most $f_2(w)$ vertices (unless {\bf shrink}$(H,B)$ reports that $I$ is a \no-instance).
\end{itemize}

In the base case, that is, when the first call {\bf shrink}$(G,\emptyset)$ makes the traversal $(0,0)\rightarrow (1,0)$, then it is clear that the statement ($a'$) holds. Now we consider the case when a traversal is made {\em to} $(i,j)$. As an induction hypothesis, we assume that any traversal to  $(i',j')$ with $i'+j' <i+j$ satisfies either ($a'$) or ($b')$, depending on whether it increase the first or the second cordinate.
%We prove the base case, i.e. that both statements $(a')$ and ($b')$ hold when $j=0$. The statement ($a'$) can be easily verified since at the initial call {\bf shrink}$(G,\emptyset)$, we have $G=G(I')$ and each subsequent update corresponds to taking an induced subgraph of $G$. Notice that the first update from $(i,0)$ to $(i,1)$ occurs at the return of {\bf shrink}$(G,\emptyset)$ corresponding to a leaf node of the recursion tree. By induction hypothesis, the graph $H$ is an induced subgraph of $G(I')$ for the current global variable $I'$ and thus, the instance $I'[V(H)]$ is well defined at line~\ref{5ufgdudgf}. Also, clearly the graph returned by {\bf shrink}$(H,B)$ is obtained by contracting edges of $H$. The number of vertices in the graph returned by {\bf shrink}$(H,B)$ is at most $f_2(w)$ since we contract all edges except for those in $\bigcup_{{\cal V}_{\bf w} \in {\cal S}}E({\cal V}_{\bf w})$, whose size is bounded by $w\cdot f_1(w)<f_2(w)$, and the obtained graph is connected. Therefore, the statement ($b'$) holds in this case. This proves the base case.

%For $j>0$, we want to settle ($a'$) and ($b'$) assuming that both statements ($a'$) and ($b'$) hold for all $(i',j')$ lexicographically smaller than $(i,j)$. There are a few possibilities.

\noindent {\bf Case A}: when the traversal $(i,j)\rightarrow (i+1,j)$ is made at Line~{\bf \ref{23vFDgdf}}: Let {\bf shrink}$(H[V_i]$,
$B')$ be the $(i+1)$-st call invoked at line~\ref{23vFDgdf}. By induction hypothesis, $H$ is an induced subgraph of $G(I')$, $H[V_i]$ is an induced subgraph of $H$, and $I'$ does not changes between $i$-th and $(i+1)$-st calls. Hence, the statement ($a'$) holds.

\noindent {\bf Case B}: when the traversal $(i,j)\rightarrow (i,j+1)$ is made at Line~{\bf \ref{23rwef}}: Clearly the graph returned by {\bf shrink}$(H,B)$ is obtained by contracting edges of $H$. The number of vertices in the graph returned by {\bf shrink}$(H,B)$ is at most $f_2(w)$ since we contract all edges except for those in $\bigcup_{{\cal V}_{\bf w} \in {\cal S}}E({\cal V}_{\bf w})$, whose size is bounded by $w\cdot f_1(w)<f_2(w)$, and the obtained graph is connected. Therefore, the statement ($b'$) holds in this case.

\noindent {\bf Case C}: when the traversal $(i,j)\rightarrow (i,j+1)$ is made at Line~{\bf \ref{wneklrw}}: Here we consider the case when {\bf shrink}$(H,B)$ returns an output at line~\ref{wneklrw}. By induction hypothesis, $H'$ is obtained by contracting edges of $H[V_i]$. Notice that $H^{\rm new}=H' \oplus_{\delta} H[V_{3-i}]$ can be obtained from $H$ by contracting the same set of edges that have been contracted in $H[V_1]$, resulting in $H'$. It clearly contains at most $f_2(w)$ vertices, thereby satisfying the statement ($b'$).

\noindent {\bf Case D}: when the traversal $(i,j)\rightarrow (i+1,j)$ is made at Line~{\bf \ref{4tgergd}}: Let $I'$ be the global variable when {\bf shrink}$(H,B)$ is called, and $I''$ be the global variable when {\bf shrink}$(H^{\rm new},B^{\rm new})$ is called at Line~{\bf \ref{4tgergd}}. By induction hypothesis, $H'$ and thus $H^{\rm new}$ can be obtained by contracting edges of $H$. During the traversal from the call {\bf shrink}$(H,B)$ and the call {\bf shrink}$(H^{\rm new},B^{\rm new})$, the contracted edges that transformed $I'$ to $I''$ are exactly those which transformed $H[V_1]$ to $H'$, and equivalently $H$ to $H^{\rm new}$. Since $H$ is an induced subgraph of $G(I')$, we conclude that $H^{\rm new}$ is an induced subgraph of $G(I'')$, and the statement ($a'$) holds.

\noindent {\bf Case E}: when the traversal $(i,j)\rightarrow (i,j+1)$ is made at Line~{\bf \ref{4tgergd}}: Here we consider the case when {\bf shrink}$(H,B)$ returns an output at Line~{\bf \ref{4tgergd}}. By case D, we know that $H^{\rm new}$ is an induced subgraph of the current global variable $I'$. By induction hypothesis, the graph returned by {\bf shrink}$(H^{\rm new},B^{\rm new})$ can be obtained by contracting edges of $H^{\rm new}$. Since $H^{\rm new}$ itself can be obtained by contraction edges of $H$ by the same argument as in Case C, the graph returned can be obtained by contracting edges of $H$. To see that the number of vertices in the returned graph is at most $f_2(w)$, we resort to the induction hypothesis, Case B and C. Therefore, ($b'$) holds as well.
\end{proofofclaim}

In order to establish the correctness of the algorithm {\bf shrink}$(H,B)$, we need the following claim.

\begin{claim}\label{claim:equiv}
While the initial call {\bf shrink}$(G,\emptyset)$ is carried out, the global variable $I'$ remains equivalent to $I$.
\end{claim}
\begin{proofofclaim}
Claim~\ref{claim:subgraph+size} implies that at Line~{\bf \ref{5ufgdudgf}}, $I'$ and $I'[V(H)]$ are indeed instances of \cla\ and \hcla\ meeting the conditions of Lemma~\ref{lem:contraction-equivalent}. Hence, any current global variable $I'$ and the new global variable $I'/E_{{\cal C}}$ updated at line~\ref{43trgdfg} are equivalent by Lemma~\ref{lem:contraction-equivalent}. As $I'=I$ at the outset of {\bf shrink}$(G,\emptyset)$, the current global variable $I'$ is equivalent to $I$ during the course of {\bf shrink}$(G,\emptyset)$.
\end{proofofclaim}

If {\bf shrink}$(G,\emptyset)$ reports that $I$ is a \no-instance, it means that $\tilde{\frak{U}}(I'[V(H)],B)=\emptyset$ for some call {\bf shrink}$(H,B)$. Note that if $I'$ is a \yes-instance, $\tilde{\frak{U}}(I'[V(H)],B)\neq \emptyset$ for every subgraph $H$ of $G(I')$. By Claim~\ref{claim:subgraph+size}, the graph $H$ is indeed a subgraph of $G(I')$ for any call {\bf shrink}$(H,B)$ incurred in the course of {\bf shrink}$(G,\emptyset)$ and thus $I'$ is a \no-instance. Together with Claim~\ref{claim:equiv}, this implies that $I$ indeed a \no-instance.

Suppose that {\bf shrink}$(G,\emptyset)$ returns a graph (i.e. does not report that $I$ is a \no-instance). By Claim~\ref{claim:subgraph+size} this means  that the graph {\bf shrink}$(G,\emptyset)$ can be obtained from $G$ by contracting edges of $G$, and has at most $f_2(w)$ vertices. Let $I^{\rm new}$ be the final global variable when {\bf shrink}$(G,\emptyset)$ terminates. The sequence of edge contractions applied to $G$ leading to {\bf shrink}$(G,\emptyset)$ is also applied to the initial global variable $I'=I$. Therefore, $G(I^{\rm new})$ contains no more vertices than the graph {\bf shrink}$(G,\emptyset)$ does, which is at most $f_2(w)$. From Claim~\ref{claim:equiv}, $I^{\rm new}$ is indeed equivalent to $I$. Therefore, by applying the algorithm of Lemma~\ref{j4oqndua} to $I^{\rm new}$, we can correctly solve the instance $I$.

What remains is to prove that \cla\ can be solved in the claimed running time, assuming an algorithm which solves {\sc HCLA}
in $f(w)\cdot p(n)$ steps. Let now $T(n,w)$ be the running time of Algorithm {\bf shrink}
when it runs on an instance $I=(G,r,\lambda,\alpha)$ where $|V(G)|=n$
and $w=\sum \alpha$.
Notice that
$$T(n,w)\leq \max_{{f_{1}(w)\leq n'\leq n-f_1(w)}}\{T_{1}(n,w)+T(n',w)+T(f_{2}(w)+n-n',w),T_{2}(n,w)\},$$
where $T_{1}$ is the running time of  required by line~{\bf \ref{k6tud8whh4}}
and $T_{2}$ is the running time required to compute $E_{{\cal C}}$ in  line~{\bf \ref{5ufgdudgf}}. From Proposition~\ref{bhtefdjfh}, $T_{1}(n,w)=2^{O(w^2\cdot  \log w)}\cdot n^3\cdot \log n$ and, from Observation~\ref{jf7sgekf73hd},
$T_{2}(n,w)=f(w)\cdot p(n)\cdot f_{1}(w)$.
By resolving the above recursion,  we obtain that
$T(n,w)=\max\{T_{1}(n,w)\cdot n,T_{2}(n,w)\}$. Lastly, solving an instance with at most $f_2(w)$ vertices using an algorithm of Lemma~\ref{j4oqndua} requires at most $2^{O(w^2\cdot \log w)}$ steps, which yields the claimed running time.
\end{proof}
%
%\IncMargin{2.5em}
%\begin{algorithm}[h]
%\LinesNumbered
%\SetNlSty{texttt}{\bf    }{.     \ \  }
%  \SetKwData{Let}{\bf let}
%  \SetKwFunction{Return}{\bf return}
%  \SetKwInOut{Name}{\bf Algorithm}\SetKwInOut{Input}{\sl Input}
%  \SetKwInOut{Output}
%  {\sl Output}
%  \Name{{\bf shrink}($I,B$)}
%  \Input{An instance $I=(G,r,\lambda,\alpha)$ of  \cla\ and a set $B\subseteq V(G)$ where $|B|\leq 2\cdot w$ and $|V(G)|> f_{2}(w)$.}
%  \Output{An instance $I^{\rm new}$ that is equivalent to $I$ whose graph has at most $f_{2}(w)$ vertices or a report that $I$ is a {\sc no}-instance.}
%  \BlankLine
%  \eIf{$G$ has a $(f_{2}(w),w)$-separation $(V_{1},V_{2})$\label{k6tud8whh4}}{
% \Let $i$ be an integer in $\{1,2\}$ such that $|B\cap V_{i}|\leq w$ \\
%% \Let $J={\delta}(V_{1},V_{2})$, $A_{i}=V_{i}\cap V(J)$, and $A_{3-i}=V_{3-i}\cap V(J)$ \\
%\Let $B'= (B\cap V_{i}) \cup (V({\delta}(V_{1},V_{2})) \cap V_{i})$ \\
% \Let $I'=I[V_{i}]$ and $I''=I[V_{3-i}]$ \\
%\Let $I^{\prime {\rm new}}$= {\bf shrink}$(I',B')$ \\
%\Let $I^{\rm new}=I^{\prime \rm new}\oplus_{\bf q} I'',$  where ${\bf q}=(A_{1},A_{2},J)$\\
%\eIf{$|V(I^{\rm new})|> f_2(w)$}{
%{\bf return} {\bf shrink}$(I^{\rm new},B)$ \\%
%}{{\bf return} $I^{\rm new}$ \\
%}
%}{
%\eIf{$\tilde{\frak{U}}(I,B)=\emptyset$}{report that $I$ is a {\sc no}-instance \\}{compute $M_{I,B}$\label{5ufgdudgf}
%\\ \Let $Q=M_{I,B}\setminus B$ \\
%{\bf return} $I^{\rm new}=I \langle Q \rangle$ \\
%}}
%\end{algorithm}
%\DecMargin{1em}

\vspace{0mm}\subsection{Split highly connected list allocation}
\label{mnop93e}

Given a graph $G$, an integer $r\in\Bbb{Z}_{\geq 1}$,  an allocation
${\cal V}=\{{\cal V}^{(1)},\ldots,{\cal V}^{(r)}\}$ of $V(G)$ and  two integers $j\in [r]$ and $x\in\Bbb{Z}_{\geq 0}$, we say that ${\cal V}$ is {\em $x$-bounded out of $j$} if
$\sum_{i\in[r]\setminus \{j\}} |{\cal V}^{(i)}|\leq x.$
We define the {\sc Split Highly Connected \allocation} problem ({\shcla}, in short) so that its instances are as the
instances of  {\sc Highly Connected \allocation}  enhanced with some subset $S$ of $V(G)$
and where we impose that $|V(G)|> 2w\cdot f_{2}(w)$ and that the requested solution ${\cal V}$,
additionally, satisfies  the following  condition: There exists some $j\in[r]$, such that\medskip

~~~~{\bf A}. ${\cal V}$ is $w\cdot f_{2}(w)$-bounded out of $j$
and

~~~~{\bf B}. $\partial_{G}({\cal V}^{(j)})\subseteq S\subseteq  {\cal V}^{(j)}$.
%\end{itemize}
%\end{itemize}

%The next lemma uses Observation~\ref{glojolg}.

\begin{lemma}
\label{opensi}
Let ${\cal V}$ be a solution of \hcla\
for an instance $I=(G,r,\lambda,\alpha)$ where $|V(G)|> 2w \cdot f_2(w)$. Then there is a unique
$j\in[r]$ such that ${\cal V}$ is $w\cdot f_2(w)$-bounded out of $j$
and
a unique $C\in {\cal C}(G\setminus E({\cal V}))$ with $|V(C)|>f_2(w)$.
Moreover, for such $C$ and $j$,  $C$ is a subgraph of $G[{\cal V}^{(j)}]$.
\end{lemma}

%%%%
\begin{proof} Let  $C$ be a connected component of  $G\setminus E({\cal V})$ that has maximum number of vertices. As $G\setminus E({\cal V})$ has at most $w+1$ connected components and $2w \cdot f_2(w) \geq (w+1)\cdot f_2(w)+1$, we deduce
that $|V(C)|> f_2(w)$. Using Observation~\ref{glojolg}, we know that $C$ belongs entirely in some ${\cal V}^{(j)}$. As  $G$ is  $(f_2(w),w+1)$-connected,
 every connected component of $G\setminus E({\cal V})$ that is different from $C$ has at most  $f_2(w)$ vertices.
 This implies that the union of the  parts of ${\cal V}$ that are different from ${\cal V}^{(j)}$ contains  at most $w\cdot f_2(w)$ vertices. Moreover $j$ is unique as, otherwise, $|V(G)|\leq 2w\cdot f_2(w)$.
\end{proof}

%The proof of the next Lemma uses Lemma~\ref{opensi}.
%
%
%\begin{lemma}
%\label{u85na1gpe}
%If $(I,S)$ is a {\sc YES}-instance of \shcla, where $I=(G,r,\lambda,\alpha)$ and $|V(G)|> 2w\cdot f_{2}(w)$,
%then there exists some solution ${\cal V}$ of \shcla\ for $(I,S)$ and a unique $j\in[r]$
%such that, {\bf i}. ${\cal V}$ is $w\cdot f_{2}(w)$-bounded out of $j$ and
%{\bf ii}.  if $C\in{\cal C}(G\setminus S)$ and
%$C$ is not ${(j,\lambda)}$-friendly,
% then $V(C)\cap  {\cal V}^{(j)}=\emptyset$.
%\end{lemma}
%
%\begin{proof} Let ${\cal V}$ be a solution of \shcla\ for $(I,S)$. From Lemma~\ref{opensi}, there is
%a unique index $j$ satisfying {\bf A}.  We also adjust the choice of ${\cal V}$ such that $|{\cal V}^{(j)}|$ is
%maximized. Obviously Condition {\bf i} holds.
%To prove Condition {\bf ii},  consider a connected component $C$ of $G\setminus S$.
%Suppose that $C$ is not ${(j,\lambda)}$-friendly. Clearly,  $V(C)$ is not a subset of ${\cal V}^{(j)}$,
%therefore it contains some vertex $x$ not in ${\cal V}^{(j)}$. Towards a contradiction we assume that
%$C$ has also a vertex $y$ in
%${\cal V}^{(j)}$. This is impossible as every path between $x$ and $y$ in $C$ should
%contain some vertex $z$ of $\partial_{G}({\cal V}^{(j)})$. From {\bf B},
%$z\in S$, a contradiction as $C$ is a connected component of $G\setminus S$.
%\end{proof}
%Before stating the next lemma, we need the following result.

\begin{proposition}[Chitnis {et al}.~\cite{ChitnisCHPP12desi}]
\label{jrif6thf}
There exists an algorithm that given a set $U$ of size
$n$ and  two integers
$a,b\in[0,n]$, outputs a set ${\cal F}\subseteq 2^{U}$
with $|{\cal F}|=2^{O(\min\{a,b\}\cdot\log(a+b+1))}\cdot \log n$ such that for every two sets $A,B\subseteq U$, where $A\cap B=\emptyset$ and $|A|\leq a$ and $|B|\leq b$, there exists a set $S\in{\cal F}$ with $A\subseteq S$ and $B\cap S=\emptyset$, in $2^{O(\min\{a,b\}\cdot\log(a+b+1))}\cdot n\cdot \log n$ steps.
\end{proposition}

The proof of the following lemma uses
Proposition~\ref{jrif6thf} and Lemmata~\ref{j4oqndua} and~\ref{opensi}.

\begin{lemma}
\label{p01j5ds}
Given  an algorithm solving
\shcla\ in $f(w)\cdot p(n)$ steps, then there is an algorithm solving
\hcla\ in $f(w)\cdot 2^{O(w^{2}\cdot \log w )}\cdot \log n\cdot  \max\{n,p(n)\}$ steps.
\end{lemma}

\begin{proof} Let $I$ be an instance of \hcla. If
$|V(G)|\leq 2w\cdot f_{2}(w)$, \hcla\ can be solved
in $ (2w\cdot f_{1}(w))^{w}\cdot 2^{O(w\cdot \log w)}=2^{O(w^2\cdot\log w )}$
steps because of Lemma~\ref{j4oqndua} (applied for $\ell=1$ and $r\leq 2w$).

Let ${\cal F}$ be a family of subsets of $V(G)$ such that the condition of Proposition~\ref{jrif6thf} is satisfied for $a=w$ and
$b=w\cdot f_{2}(w)$. We claim that $I$ is a {\sc yes}-instance of
\hcla\ if and only if for some $S\in{\cal F}$, $(I,S)$ is a {\sc yes}-instance of \shcla. Recall that $(I,S)$ is an instance of \shcla, as $|V(G)|>2w\cdot f_{2}(w)$.

In the non-trivial direction, assume that ${\cal V}$ is a solution for $I$. By applying Lemma~\ref{opensi} on $I$, we know that there is a unique $j$ such that ${\cal V}$ is $w\cdot f_{2}(w)$-bounded out of $j$.
Let
$A=\partial_{G}({\cal V}^{(j)}) \mbox{\ \ and\ \ } B=\bigcup_{i\in[r]\setminus\{j\}}{\cal V}^{(i)}.$
Clearly, $|A|\leq w$
and $|B|\leq w\cdot f_{2}(w)$.
By the definition of ${\cal F},$ there exists some set $S\in{\cal F}$ such that $A\subseteq S$ and $B\cap S=\emptyset$. Therefore $\partial_{G}({\cal V}^{(j)})\subseteq S\subseteq {\cal V}^{(j)}$ and $(I,S)$ is a {\sc yes}-instance of \shcla\  as required.

Suppose now that ${\sf A}$ is an algorithm that solves \shcla\ in $f(w)\cdot p(n)$ steps. To solve \hcla, we  apply {\sf A} on $(I,S)$ for all $S\in{\cal F}$. If we obtain a solution to $(I,S)$ for some $S\in{\cal F}$ we output this solution as a solution to $I$, otherwise we output that $I$ is a {\sc no}-instance of \hcla.
 As $|{\cal F}|=2^{O(w\cdot\log(w\cdot f_{2}(w)))}\cdot \log n=2^{O(w^{2}\cdot \log w )}\cdot \log n$, this algorithm runs in $2^{O(w^{2}\cdot \log w )}\cdot \log n\cdot n+2^{O(w^{2}\cdot \log w )}\cdot \log n\cdot f(w)\cdot p(n)$ steps as required.
\end{proof}

%Let $(I,S)$ be an instance of {\sc SHCLA} where $I=(G,r,\lambda,\alpha)$ and $S\subseteq V(G)$,
%and let ${\cal W}\subseteq {\cal C}(G\setminus S)$ and $j\in[r]$.
%We define the function $\lambda_{j,S}: V(G)\rightarrow 2^{[r]}$ such that
%$$\lambda_{j,S}(x)=\left\{\begin{array}{lll}
%\{j\}& x\in S, \\
%\lambda(x)\setminus\{j\}&  x\in V(G)\setminus S.
%\end{array}\right.$$

\vspace{0mm}\subsection{An algorithm for solving \shcla}
\label{mnr34op973e}

Below we present a dynamic programming algorithm for solving \shcla.% and Observation~\ref{bgbg3ertu6}.

\begin{lemma}
\label{opp4pp4kjj4j}
\shcla\ can be solved in $2^{O(w^{2}\cdot \log w )}\cdot n$ steps.
\end{lemma}
\begin{proof}
We present a dynamic programming for \shcla\ using the brute-force algorithm of Lemma~\ref{j4oqndua} as a subroutine.
Let $(I,S)$ be an instance of {\sc SHCLA} where $I=(G,r,\lambda,\alpha)$ and $S\subseteq V(G)$,
and let $C_1,\ldots , C_{\ell}$ be the vertex sets of the graphs in ${\cal C}(G\setminus S)$. Let $G_i=G[S \cup \bigcup_{1\leq i'\leq i} C_{i'}]$ for $1\leq i\leq \ell$ and specifically $G_0=G[S]$.
For $s\in[r]$, we define two functions $\lambda_{s}, \lambda_{s}^*: V(G)\rightarrow 2^{[r]}$ such that

$$\lambda_{s}(x)=\left\{\begin{array}{lll}
\{s\}& x\in S \\
\lambda(x)&  x\in V(G)\setminus S
\end{array}\right.$$ and
$$\lambda_{s}^*(x)=\left\{\begin{array}{lll}
\{s\}& x\in S, \\
\lambda(x)\setminus \{s\} &  x\in V(G)\setminus S.
\end{array}\right.$$

For each $s\in [r]$, we have a table $P_s$ for dynamic programming in which the entries $P_s(i,\alpha',c')$ are either \yes\ or \no, and run over all $0\leq i\leq \ell$, $\alpha'\in\frak{F}_{\leq}(\alpha)$ and $0\leq c'\leq w\cdot f_{2}(w)$. The entries of $P_s$ are determined recursively as follows.

\begin{itemize}
\item $P_s(0,\alpha',c')=\yes$ if and only if $\alpha'=\mathbf{0}$, $c'=0$ and $G[S]$ is $(s,\lambda)$-friendly.
\item For $1\leq i \leq \ell$, $P_s(i,\alpha',c')=\yes$ if and only if
\begin{itemize}
\item[(i)] $P_s(i-1,\alpha',c')=\yes$ and $G[C_i]$ is $(s,\lambda)$-friendly, or
\item[(ii)] $|C_i|\leq c'$, and there exists $\alpha''\in \frak{F}_{\leq}(\alpha')$ such that $P_s(i-1,\alpha'',c'-|C_i|)=\yes$ and $(G[S\cup C_i],r,\lambda_s^*|_{S\cup C_i},\alpha'-\alpha'')$ is a \yes-instance for \la.
\end{itemize}
\end{itemize}
%We set $P_s(0,\alpha',c')=\yes$ if and only if $\alpha'=\mathbf{0}$, $c'=0$ and $G[S]$ is $(s,\lambda)$-friendly. When computing $P_s(i,\alpha',c')$, we can assume that all values $P(j,\alpha'',c'')$ have been determined for $j<i$, $\alpha'' \in\frak{F}_{\leq}(\alpha)$ and $0\leq c''\leq w\cdot f_{2}(w)$. For $i\geq 1$, we set $P_s(i,\alpha',c')=\yes$ if and only if
%(i) $P_s(i-1,\alpha',c')=\yes$ and $G[C_i]$ is $(s,\lambda)$-friendly, or (ii) $|C_i|\leq c'$, $P_s(i-1,\alpha'',c'-|C_i|)=\yes$ and $(G[S\cup C_i],r,\lambda_s^*|_{S\cup C_i},\alpha'-\alpha'')$ is a \yes-instance for \la.

\begin{claim}\label{claim:iff1}
$P_s(i,\alpha',c')=\yes$ if and only if the instance $(G_i,r,\lambda_s|_{V(G_i)},\alpha')$ admits a solution ${\cal V}$ for \la\ such that $\sum_{i\in[r]\setminus \{s\}} |{\cal V}^{(i)}|=c'$, and either $C_j \subseteq {\cal V}^{(s)}$ or $C_j \cap {\cal V}^{(s)}=\emptyset$ for each $1\leq j\leq i$.
\end{claim}
\begin{proofofclaim}
When $i=0$, it is tedious to verify the claim. We prove by induction on $i$.

Firstly, we prove the forward direction. Suppose that $P_s(i,\alpha',c')=\yes$ and consider the instance $(G_i,r,\lambda_s|_{V(G_i)},\alpha')$ of \la. If case (i) of the recursion holds, then by induction hypothesis, there exists a solution ${\cal V}'$ to $(G_{i-1},r,\lambda_s|_{V(G_{i-1})},\alpha')=(G_i\setminus C_i,r,\lambda_s|_{V(G_i)\setminus C_i},\alpha')$ such that $\sum_{i\in[r]\setminus \{s\}} |{\cal V}'^{(i)}|=c'$, and either $C_j \subseteq {\cal V}^{(s)}$ or $C_j \cap {\cal V}^{(s)}=\emptyset$ for each $1\leq j\leq i-1$. Let ${\cal V}$ be an $r$-allocation obtained from ${\cal V}'$ by adding all vertices of $C_i$ to the part ${\cal V}^{(s)}$. Since $G[C_i]$ is $(s,\lambda)$-friendly and $C_i\cap \partial_{G_i}({\cal V}^{(s)})=\emptyset$, it follows that ${\cal V}$ is a solution to $(G_i,r,\lambda_s|_{V(G_i)},\alpha')$. Note that ${\cal V}$ meets the two conditions of our claim.\medskip

Suppose that case (i) of the recursion does not hold for the entry $P_s(i,\alpha',c')=\yes$, but case (ii) does. From $P_s(i-1,\alpha'',c'-|C_i|)=\yes$ and induction hypothesis, there exists a solution ${\cal V}'$ to $(G_{i-1},r,\lambda_s|_{V(G_{i-1})},\alpha'')=(G_i\setminus C_i,r,\lambda_s|_{V(G_i)\setminus C_i},\alpha'')$ such that $\sum_{i\in[r]\setminus \{s\}} |{\cal V}^{(i)}|=c'-|C_i|$, and either $C_j \subseteq {\cal V}^{(s)}$ or $C_j \cap {\cal V}^{(s)}=\emptyset$ for each $1\leq j\leq i-1$. Let ${\cal V}''$ be a solution to $(G[S\cup C_i],r,\lambda_s^*|_{S\cup C_i},\alpha'-\alpha'')$, and let ${\cal V}$ be ${\cal V}'\cup {\cal V}''$. Indeed, ${\cal V}$ is an $r$-allocation of $V(G_i)$ since $V(G_{i-1})\cap (S\cup C_i)=S$, $S\subseteq {\cal V}'^{(s)}$ and $S\subseteq {\cal V}''^{(s)}$. It is easy to see that ${\cal V}$ is a solution to $(G_i,r,\lambda_s|_{V(G_i)},\alpha')$. Furthermore, due to the definition of $\lambda^*_s$, we have $\sum_{i\in[r]\setminus \{s\}}|{\cal V}''^{(i)}|=|C_i|$, and thus $\sum_{i\in[r]\setminus \{s\}} |{\cal V}^{(i)}|=\sum_{i\in[r]\setminus \{s\}} |{\cal V}'^{(i)}|+ \sum_{i\in[r]\setminus \{s\}}|{\cal V}''^{(i)}|=(c'-|C_i|)+|C_i|=c'$. Notice that ${\cal V}^{(s)}={\cal V}'^{(s)}$ as $C_i\cap {\cal V}''^{(s)}=\emptyset$. This implies that we have $C_j \subseteq {\cal V}'^{(s)}={\cal V}^{(s)}$ or $C_j \cap {\cal V}^{(s)}=C_j\cap {\cal V}'^{(s)}=\emptyset$ for $1\leq j \leq i-1$. It remains to observe that $C_i\cap {\cal V}''^{(s)}=\emptyset$ also implies $C_i \cap {\cal V}^{(s)}=\emptyset$.

Secondly, let us prove the opposite direction. Let ${\cal V}$ be a solution to $(G_i,r$,
$\lambda_s|_{V(G_i)},\alpha')$ meeting the conditions of the claim. Consider the two cases. \medskip

\noindent Case 1: Suppose $C_i \subseteq {\cal V}^{(s)}$. Clearly, $G[C_i]$ is $(s,\lambda)$-friendly. We argue that $P_s(i-1,\alpha',c')=\yes$, which implies $P_s(i,\alpha',c')=\yes$ by the recursion case (i) for $P_s$. Let ${\cal V}'$ be the restriction of ${\cal V}$ to $V(G_{i-1})$, i.e. $({\cal V}^{(1)}\setminus C_i, \ldots , {\cal V}^{(r)}\setminus C_i)$. By induction hypothesis, in order to prove $P_s(i-1,\alpha',c')=\yes$, it suffices to show that ${\cal V}'$ is a solution to $(G_{i-1},r,\lambda_s|_{V(G_{i-1})},\alpha')$ such that $\sum_{j\in[r]\setminus \{s\}} |{\cal V}'^{(j)}|=c'$, and either $C_j \subseteq {\cal V}^{(s)}$ or $C_j \cap {\cal V}^{(s)}=\emptyset$ for each $1\leq j\leq i-1$. Indeed, ${\cal V}'$ is a solution to $(G_{i-1},r,\lambda_s|_{V(G_{i-1})},\alpha')$ for $C_i \subseteq {\cal V}^{(s)}\setminus \partial_{G_i}({\cal V}^{(s)})$, which implies $\delta_{G_{i-1}}({\cal V}'^{(j)}, {\cal V}'^{(k)})=\delta_{G_i}({\cal V}^{(j)}, {\cal V}^{(k)})$ for every $1\leq j<k \leq r$.
%\[|\delta_{G_{i-1}}({\cal V}'^{(j)}, {\cal V}'^{(k)})| = \left\{\begin{array}{lll}
%|\delta_{G_i}({\cal V}^{(j)}, {\cal V}^{(k)})|=\alpha'(j,k) & \text{if~} s\neq j,k \\
%|\delta_{G_{i-1}}({\cal V}'^{(s)}, {\cal V}'^{(k)})|=|\delta_{G_i}(\partial({\cal V}^{(s)}), {\cal V}^{(k)})|=\alpha'(s,k) & \text{if~} j=s\\
%|\delta_{G_{i-1}}({\cal V}'^{(j)}, {\cal V}'^{(s)})|=|\delta_{G_i}({\cal V}^{(j)},\partial({\cal V}^{(s)}) )|=\alpha'(j,s) & \text{if~} k=s\\
%\end{array}\right.\]
Note that $\sum_{j\in[r]\setminus \{s\}} |{\cal V}'^{(j)}|=\sum_{j\in[r]\setminus \{s\}} |{\cal V}^{(j)}|=c'$. Moreover, from ${\cal V}'^{(s)}={\cal V}^{(s)}\setminus C_i$, ${\cal V}'^{(j)}={\cal V}^{(j)}$ for $j\neq s$, and the fact that
either $C_j \subseteq {\cal V}^{(s)}={\cal V}'^{(s)} \cup C_i$ or $C_j \cap {\cal V}^{(s)}=\emptyset$ holds for each $1\leq j\leq i-1$, we have either $C_j \subseteq {\cal V}'^{(s)}$ or $C_j \cap {\cal V}'^{(s)}=C_j \cap ({\cal V}^{(s)}\setminus C_i)= C_j \cap  {\cal V}^{(s)}=\emptyset$ for each $1\leq j\leq i-1$. \medskip

\noindent Case 2: Suppose $C_i \cap {\cal V}^{(s)}=\emptyset$. For $\sum_{j\in[r]\setminus \{s\}} |{\cal V}^{(j)}|=c'$, we have $|C_i|\leq c'$. Let ${\cal V}'$ be the restriction of ${\cal V}$ to $V(G_{i-1})$, i.e. $({\cal V}^{(1)}\setminus C_i, \ldots , {\cal V}^{(r)}\setminus C_i)$ and let ${\cal V}''$ be the restriction of ${\cal V}$ to $S\cup C_i$. Also let $\alpha''$ be such that $\alpha''(j,k)=|\delta_{G_{i-1}}({\cal V}'^{(j)}, {\cal V}'^{(k)})|$ for every $1\leq j < k\leq r$. In order to show $P_s(i,\alpha',c')=\yes$, it suffices to verify that $P_s(i-1,\alpha'',c'-|C_i|)=\yes$ and $(G[S\cup C_i],r,\lambda_s^*|_{S\cup C_i},\alpha'-\alpha'')$ is a \yes-instance for \la. \medskip

To verify $P_s(i-1,\alpha'',c'-|C_i|)=\yes$, notice that ${\cal V}'$ is a solution to $(G_{i-1},r$,
$\lambda_s|_{V(G_i)}$,
$\alpha'')$ and $\sum_{i\in[r]\setminus \{s\}} |{\cal V}'^{(i)}|=\sum_{i\in[r]\setminus \{s\}} |{\cal V}^{(i)}\setminus C_i|=\sum_{i\in[r]\setminus \{s\}} |{\cal V}^{(i)}|-|C_i|=c'-|C_i|$. For each $1\leq j\leq i-1$, if $C_j \subseteq {\cal V}^{(s)}$, then $C_j \subseteq {\cal V}'^{(s)}$ since ${\cal V}^{(s)}={\cal V}'^{(s)}\cup C_i$ and $C_j\cap C_i=\emptyset$. Otherwise, $C_j \cap {\cal V}'^{(s)}=C_j \cap {\cal V}^{(s)}=\emptyset$ for ${\cal V}'^{(s)}={\cal V}^{(s)}$. By induction hypothesis, that $P_s(i-1,\alpha'',c'-|C_i|)=\yes$ follows. It is routine to check that ${\cal V}''$ is a solution to $(G[S\cup C_i],r,\lambda_s^*|_{S\cup C_i},\alpha'-\alpha'')$ is a \yes-instance for \la.
\end{proofofclaim}

The next claim, together with Claim~\ref{claim:iff1}, asserts that we can correctly solve \shcla\ by computing the tables $P_s$ for $s\in [r]$.

\begin{claim}\label{claim:iff2}
The given instance $(I,S)$ is \yes\ to \shcla\ if and only if $P_s(\ell,\alpha,c)=\yes$ for some $s\in [r]$ and $c\leq w\cdot f_{2}(w)$.
\end{claim}
\begin{proofofclaim}
To see the forward direction, let ${\cal V}$ be a solution to $(I,S)$ of \shcla. By definition of \shcla\,  there exists an index $j\in[r]$ such that the two conditions {\bf A}. ${\cal V}$ is $w\cdot f_{2}(w)$-bounded out of $s$ and {\bf B}. $\partial_{G}({\cal V}^{(s)})\subseteq S\subseteq  {\cal V}^{(s)}$ are met. Let $c:=\sum_{i\in[r]\setminus \{s\}}$. Notice that ${\cal V}$ is a solution to the instance $(G,r,\lambda_{s},\alpha)$ of \la\ as $S\subseteq {\cal V}^{(s)}$, and $c\leq w\cdot f_2(w)$ by Condition {\bf A.}. Hence, to prove that $P_s(\ell, \alpha,c)=\yes$, it suffices to verify that $C_j \subseteq {\cal V}^{(s)}$ or $C_j \cap {\cal V}^{(s)}=\emptyset$ holds for every $1\leq j\leq \ell$ by Claim~\ref{claim:iff1}. Suppose that $C_j \nsubseteq {\cal V}^{(s)}$ and $C_j \cap {\cal V}^{(s)}\neq \emptyset$ for some $j$. Then $C_j \cap {\cal V}^{(s)}$ constains a vertex of $\partial_G({\cal V}^{(s)})$, and thus contains a vertex of $S$ by {\bf B.}. However, this contradicts the fact that $C_j$ and $S$ are disjoint.

For the backward implication, suppose $P_s(\ell,\alpha,c)=\yes$ for some $s\in [r]$ and $c\leq w\cdot f_2(w)$. Then, there exists a solution ${\cal V}$ to the instance $(G_{\ell},r,\lambda_s,\alpha)$ satisfying the condition of Claim~\ref{claim:iff1}. It suffices to show that $\partial_{G}({\cal V}^{(s)})\subseteq S\subseteq  {\cal V}^{(s)}$. If $\partial_{G}({\cal V}^{(s)})\setminus S\neq \emptyset$, this means that there exists a vertex sets $C$ of ${\cal C}(G\setminus S)$ such that $C\setminus {\cal V}^{(s)}\neq \emptyset$ and $C\cap {\cal V}^{(s)}\neq \emptyset$, a contradiction to the second condition of Claim~\ref{claim:iff1}. The fact that $S\subseteq {\cal V}^{(s)}$ is an immediate consequence of the definition of $\lambda_s$.
\end{proofofclaim}

In the recursion for $P_s$, verifying (i) takes $O(|C_i|)$ steps and verifying (ii) amounts to solving an instance to \la\ whose instance size is at most $c\leq w\cdot f_2(w)$. The latter takes $2^{O(w^2 \cdot \log w)}$ using the algorithm of Lemma~\ref{j4oqndua}. As the size of each table $P_s$ is $(\ell +1)\cdot |\frak{F}_{\leq}(\alpha)|\cdot w\cdot f_2(w)$, we obtain the claimed running time.
\end{proof}

Composing the running times of  Lemmata~\ref{d8h4kjks},~\ref{u83nd01nfbsyr},~\ref{p01j5ds}, and~\ref{opp4pp4kjj4j} and the fact that $f_{1}(w)=2^{O(w\cdot\log w)}$ and $f_{2}(w)=2^{O(w\cdot\log w)}$ we can derive the correctness of Theorem~\ref{u7u33ejk}.

\section{Further research}
\label{sec:conclusions}

In the definition of {\sc  List Allocation} we ask for a  {\em $\lambda$-{list} $H$-homomorphism} of $G$  where $ \sum_{e\in E(H)}|C(e)|\leq \ell$. A different parameterization of   {\sc List Allocation}, that is  similar in flavor
to  \textsc{Min-Max Multiway Cut}, may instead ask for a  {\em $\lambda$-{list} $H$-homomorphism} of $G$  where $ \max_{v\in V(H)}\sum_{{\mbox{\small  $e$
is incident to $v$}}}|C(e)|\leq \ell$.  We call this new problem {\sc Max Bounded List Digraph Homomorphism} (in short {\sc MBLDH})
As it is straightforward to prove an analogue of Theorem~\ref{kl6iq2zmn6}, where  {\sc BLDH} is now replaced
by  {\sc MBLDH} and instead of $2^{O(\ell\log h)}\cdot T(n,\ell)$ steps we now have
a reduction that takes  $2^{O(\ell^2\log h)}\cdot T(n,\ell)$ steps. This implies that   {\sc MBLDH}, when parameterized
by $\ell$ and $h$  admits an {\sf FPT}-algorithm that runs in $2^{O(\ell^2\cdot\max \{\log\ell,\log h\})}\cdot n^{4}\cdot \log n$ steps.

A  natural research direction is to improve the running time of our {\sf FPT}-algorithms for \textsc{Min-Max Multiway Cut} and  \textsc{Bounded List Digraph Homomorphism}.
If we want to improve our running times using the techniques used in this paper  it seems
that we need to crucially improve upon the recursive understanding and
randomized contractions technique.

\medskip

\noindent\textbf{Acknowledgement}. We would like to thank the anonymous referees of an earlier version of this paper for their thorough remarks and suggestions that improved the presentation and some proofs of the paper.

%
%\ig{The format of the references is HIGHLY heterogeneous, we need to fix that. I suggest to put all the ``clean" references in the file ``bib-clean.bib''}
%
%{\footnotesize

%
%\sed{\small Reviewer says:  I think you should
%check the references to make sure that the formatting and
%abbreviations are consistent}

%\bibliographystyle{plain}

%\bibliographystyle{abbrv}
%\bibliography{bib-clean-Ignasi}

\end{document}